\newcount\Comments  
\documentclass[a4paper,UKenglish,cleveref, autoref, thm-restate]{lipics-v2021}

\usepackage{tcolorbox} 
\usepackage{tablefootnote}
\usepackage{graphicx}
\usepackage{textcomp, gensymb}
\usepackage{amsmath, amssymb}

\usepackage[noend, ruled, lined, linesnumbered, commentsnumbered,vlined, longend]{algorithm2e}
\usepackage{xcolor}
\usepackage[normalem]{ulem}
\usepackage{soul}

\newtheorem*{theorem*}{Theorem} 
\newtheorem*{lemma*}{Lemma}
\newtheorem*{corollary*}{Corollary}

\newtheoremstyle{repeated}{}{}{\itshape}{}{\bfseries}{.}{.5em}{#1\thmnote{ #3}}
\theoremstyle{repeated}

\SetCommentSty{mycommfont}
\definecolor{darkgreen}{rgb}{0,0.6,0}
\definecolor{purple}{rgb}{1,0,1}

\definecolor{darkgreen}{rgb}{0,0.6,0}
\definecolor{purple}{rgb}{1,0,1}

\newcommand{\length}[1]{\ell(#1)}

\usepackage{todonotes}

\newcommand{\etodo}[2]{\ifnum\Comments=1{\todo[color=#1!25, size=\tiny]{#2}}\fi}

\title{Efficient Distributed Algorithms for Shape Reduction via Reconfigurable Circuits} 

\titlerunning{Efficient Distributed Algorithms for Shape Reduction via Reconfigurable Circuits} 

\author{Nada Almalki}{Department of Computer Science, University of Liverpool, UK}{n.almalki@liverpool.ac.uk}{https://orcid.org/0000-0001-9403-1702}{}

\author{Siddharth Gupta}{Department of Computer Science \& Information Systems, BITS Pilani Goa Campus, India}{siddharthg@goa.bits-pilani.ac.in}{https://orcid.org/0000-0003-4671-9822}{}

\author{Othon Michail}{Department of Computer Science, University of Liverpool, UK}{othon.michail@liverpool.ac.uk}{https://orcid.org/0000-0002-6234-3960}{}

\author{Andreas Padalkin}{Paderborn University, Germany}{andreas.padalkin@upb.de}{https://orcid.org/0000-0002-4601-9597}{This author was supported by the DFG Project SCHE 1592/10-1.}

\authorrunning{N. Almalki, S. Gupta, O. Michail, and A. Padalkin}

\Copyright{Almalki, Gupta, Michail, Padalkin}


\usepackage{hyperref}
\EventEditors{}
\EventNoEds{2}
\EventLongTitle{}
\EventShortTitle{}
\EventAcronym{}
\EventYear{}
\EventDate{}
\EventLocation{}
\EventLogo{}
\SeriesVolume{}
\ArticleNo{}

\nolinenumbers
\pdfoutput=1
\hideLIPIcs

\begin{document}

\maketitle

\begin{abstract}
Autonomous reconfiguration of agent-based systems is a key challenge in the study of programmable matter, distributed robotics, and molecular self-assembly. While substantial prior work has focused on \emph{size-preserving} transformations, much less is known about \emph{size-changing} transformations. Such transformations find application in natural processes, active self-assembly, and dynamical systems, where structures may evolve through the addition or removal of components controlled by local rules. In this paper, we study efficient distributed algorithms for transforming 2D geometric configurations of simple agents, called \emph{shapes}, using only local size-changing \emph{operations}. A novelty of our approach is the use of \emph{reconfigurable circuits} as the underlying communication model, a recently proposed model enabling instant node-to-node communication via primitive signals. Unlike previous work, we integrate collision avoidance as a core responsibility of the distributed algorithm.

We consider two graph update models: \emph{connectivity} and \emph{adjacency}. Let $n$ denote the number of agents and $k$ the number of turning points in the initial shape. In the connectivity model, we show that any tree-shaped configuration can be reduced to a single agent using only \emph{shrinking} operations in $O(k \log n)$ rounds w.h.p., and to its \emph{incompressible} form in $O(\log n)$ rounds w.h.p. given prior knowledge of the incompressible nodes, or in $O(k \log n)$ rounds otherwise. When both \emph{shrinking} and \emph{growth} operations are available, we give an algorithm that transforms any tree to a topologically equivalent one in $O(k \log n + \log^2 n)$ rounds w.h.p. On the negative side, we show that one cannot hope for $o(\log^2 n)$-round transformations for all shapes of $\Theta(\log n)$ turning points. In the adjacency model, we show that any connected shape can reduce itself to a single node using only shrinking in $O(\log n)$ rounds w.h.p.

\end{abstract}

\section{Introduction}\label{sec:intro}
Reconfiguration and autonomous formation of networks and geometric structures is an active area of research, with applications ranging from swarm and reconfigurable robotics to dynamic networks and self-assembly. A central problem in 
this area of reconfigurable systems 
is for a set of agents to transform from an initial into a target arrangement efficiently, and without violating structural and other constraints.

Much of the existing work has focused on \emph{size-preserving} transformations, which guide the system through a sequence of configuration changes without altering its overall size. Other work, particularly in algorithmic self-assembly and distributed computing, has studied passive \emph{size-changing} dynamics, whose accumulation leads to structural changes over time. Recent research has begun to investigate transformations that can actively form or deform a configuration by locally controlling the addition or removal of individual agents. These processes are motivated by biological and chemical systems, such as embryonic growth in multicellular organisms, and by dynamical systems, like cellular and self-reproducing automata. They are also related to recent work on algorithmic reconfiguration of graphs and networks~\cite{DBLP:journals/ccr/AvinS18,DBLP:journals/dc/GotteHSW23,DBLP:journals/dc/MichailSS22,DBLP:journals/jcss/MertziosMSST25}.

To the best of our knowledge, the paper by Woods \emph{et al.}~\cite{DBLP:conf/innovations/WoodsCGDWY13} is the only study of a distributed model (called the \emph{nubot} model) for exponentially fast size-changing transformations. In terms of its dynamics, their model uses a combination of size-changing and size-preserving rules. An example of the latter is local rotations causing relative movement between agents. The application of rules that would cause collisions or other structural violations is effectively excluded by the model. As an active self-assembly model, it also makes a set of specific molecular modeling assumptions, such as its evolution over time being a continuous-time Markov process (as in stochastic CRNs) and Brownian motion-style agitation. 

Motivated by \cite{DBLP:conf/innovations/WoodsCGDWY13} and \cite{DBLP:journals/jcss/MertziosMSST25}, Almalki and Michail \cite{DBLP:journals/tcs/AlmalkiM24} proposed a centralized geometric model that is more abstract than \cite{DBLP:conf/innovations/WoodsCGDWY13}, aiming to understand what transformations can or cannot be performed (with the main focus being on polylog time) if \emph{only} size-changing rules are available. They studied restricted types of growth dynamics, conditioned to affect specific parts of the shape (e.g., whole columns) and defined in a way that ensures no collisions occur. Almalki \emph{et al.}~\cite{DBLP:conf/algosensors/AlmalkiGM24} generalized this to parallel growth rules that can be applied to any subset of the agents. These are not guaranteed to be collision-free \emph{a priori}, so collision avoidance becomes the responsibility of the algorithm designer.  
They show that some classes of shapes can be grown efficiently (in polylogarithmic time), while others require superpolylogarithmic time.

In this work, we take this investigation further, to a fully distributed setting, remaining focused on understanding the structural and (distributed) algorithmic properties of size-changing dynamics. Communication is a central challenge in taking this step. Indeed, global communication based on classical distributed models such as message passing is exponentially slower than the speed of parallel reconfiguration available, thus cannot be useful. Other issues include the lack of common orientation or chirality, and additional constraints stemming from the decentralization and limited computation and communication resources available to the agents. To address these challenges, we make use, for the first time in this class of problems, of the \emph{reconfigurable circuit} \cite{DBLP:journals/jcb/FeldmannPSD22} communication model.

Parallel transformation processes can be exponentially fast and are highly dynamic, which makes it difficult to avoid structural violations like collisions \cite{DBLP:conf/algosensors/GuptaKMP24}.
It is therefore important to have fast and reliable communication within the system---especially over long distances.

In order to overcome these issues,
Feldmann \emph{et al.} \cite{DBLP:journals/jcb/FeldmannPSD22} proposed the deployment of \emph{reconfigurable circuits} which allow to connect subsets of agents. Each agent can send simple signals (beeps) through any circuit it is connected to, which are instantaneously received by all agents connected to the same circuit. Reconfigurable circuits have been shown to enable polylogarithmic-time solutions to various problems, such as shape recognition/containment \cite{DBLP:journals/corr/abs-2501-16892,DBLP:journals/jcb/FeldmannPSD22}, spanning tree \cite{DBLP:journals/nc/PadalkinSW24}, and shortest path \cite{DBLP:conf/podc/PadalkinS24}. Reconfigurable circuits have so far only been used in the context of polylog-time distributed computations (primarily in the \emph{amoebot} model of programmable matter and for graph problems) and size-preserving reconfiguration, and have not yet been applied to exponentially fast parallel reconfiguration.

In this work, we leverage reconfigurable circuits as a communication model for rapid distributed size-changing transformations. In particular, we study distributed algorithms that must efficiently reduce a given shape into a target shape (sometimes a single node) through shrinking or by combining shrinking and growth, while avoiding collisions. This also offers a new application domain for reconfigurable circuits.

\subsection{Contribution}\label{subsec:cont}
We aim to understand the conditions under which a set of agents (also called \emph{nodes} throughout) can efficiently 
reduce their geometric arrangement to another such arrangement (referred to as \emph{shape} throughout), using either only \emph{shrinking} operations or a combination of \emph{shrinking} and \emph{growth} operations. A set of \emph{shrinking} operations reduces a shape by absorbing nodes into neighboring nodes, while a set of \emph{growth} operations expands a shape by adding nodes adjacent to existing nodes.\footnote{As will become clear, shrinking in this work involves translation, making it fundamentally different from the shrinking rule of \cite{DBLP:conf/innovations/WoodsCGDWY13} which is merely node deletion.} Our transformations result from distributed algorithms executed by limited agents, each of which operates with constant memory and is computationally equivalent to a finite-state machine. For communication, we employ the \emph{reconfigurable circuit} model, as its speed can support the fast dynamics of the transformations under consideration. Reconfigurable circuits enable efficient global coordination by allowing the nodes to exchange limited information instantly (in constant time) across a dynamically changing structure.

The challenge lies in designing efficient distributed algorithms that perform these operations while maintaining structural integrity and avoiding collisions---undesirable structural violations where nodes overlap, intersect, or whose movement breaks the structure. 
Understanding the distributed complexity and developing distributed algorithms for size-changing transformation tasks is important for engineering self-deploying systems and to begin understanding some of the algorithmic properties involved in formation processes of nature. Our work is also relevant to research on extremely weak models of distributed computing, such as population protocols~\cite{DBLP:journals/dc/AngluinADFP06} and variants, beeping models~\cite{DBLP:conf/wdag/CornejoK10}, and programmable matter~\cite{DBLP:conf/spaa/DerakhshandehDGRSS14}. In contrast to the considerable amount of work on size-preserving dynamics (e.g., in reconfigurable robotics), algorithmically controlled size-changing dynamics are less understood and the development of appropriate models is still in its infancy. We aim to take a step in this direction by proposing a novel such model and initiating its study.

We consider connected geometric shapes formed by nodes in a two-dimensional square grid in the \emph{connectivity} graph model, in which edges form only when new nodes are created.
Our goal is to determine complexity bounds and design distributed algorithms for different reduction objectives, aiming to optimize time efficiency without violating structural constraints. Our main results is the \emph{BFS shrinking} algorithm which reduces any tree to a single node in $O(k \log n)$ rounds w.h.p.\footnote{We define with high probability (w.h.p.) as a probability of success of at least $1-1/n^c$, where $c$ is a constant that can be made arbitrarily large.}, where the parameter $k$ represents the number of turning points and $n$ the number of nodes in the initial shape, and the \emph{incompressible tree} algorithm, which reduces any tree to its incompressible form in $O(\log n)$ rounds w.h.p.\ when the incompressible nodes are known in advance, or in $O(k \log n)$ rounds w.h.p.\ otherwise. Moreover, by combining these approaches, we can shrink any tree to a single node through its incompressible tree in $O( \log n+ k\log k)$ rounds. These algorithms use only shrinking operations. The \emph{target tree} algorithm, by combining shrinking and growth operations, ensures the reduction of any tree to any \emph{topologically equivalent} tree in $O(k \log n+ \log^2 n)$ rounds w.h.p. We also give centralized lower bounds (implying equivalent distributed lower bounds): ($i$) There exists an infinite family of pairs of \emph{geometrically equivalent} paths of $k = \Theta(\sqrt[3]{n})$ turning points each such that, if only shrinking is available, $\Omega(k \log k)$ rounds are required to reduce one path to the other.
($ii$) There exists an infinite family of pairs of \emph{geometrically equivalent} paths of $k = \Theta(\log n)$ turning points each such that, if only shrinking is available, $\Omega(\log^2 n)$ rounds are required to reduce one path to the other. In the \emph{adjacency} model, the \emph{shape reduction} algorithm shrinks any shape to a single node in $O(\log n)$ rounds w.h.p. These results are summarized in~\Cref{tab:results}.

\begin{table}[tbp]
    \centering
    \resizebox{\textwidth}{!}{ 
    \begin{tabular}{l|c|
    c}
    
        \hline
        \textbf{Reduction to} & \textbf{Model} 
        & \textbf{Bound} \\
        \hline
        \textit{Single node} & Connectivity Model
        & $O(k \log n)$ w.h.p. \\
        \textit{Single node} 
        (known incompressible nodes) & --
        & $O(\log n + k \log k)$ w.h.p. \\
        \textit{Geometrically equivalent tree} ($k = \Theta(\sqrt[3]{n})$) & --  
        & $\Omega(k\log k)$ \\
        \textit{Geometrically equivalent tree} ($k = \Theta(\log n)$) & --  
        & $\Omega(\log^2 n)$ \\
        \textit{Incompressible tree}  
        & -- 
        & $O(k \log n)$ \\
        \textit{Incompressible tree} (known incompressible nodes) & -- 
        & $O(\log n)$ \\
        \textit{Topologically equivalent tree}  & -- 
        & $O(k \log n + \log^2 n)$ \\
        \hline
         \textit{Single node} & Adjacency Model
         &$O(\log n)$ \\
         \hline
    \end{tabular}
    }
    
    \caption{
        A summary of our results.
        All upper bounds require a preprocessing phase that requires $O(\log n)$ rounds w.h.p.\ that we omit in the bounds.
    }
    \label{tab:results}
\end{table}

The remainder of the paper is organized as follows.
In \Cref{sec:model}, we formally introduce our \emph{distributed} computational model and define the growth and shrinking operations.
In \Cref{sec:problem}, we define the \emph{reduction} problem and various variants that we consider in this work. In \Cref{sec:prelim}, we present key algorithmic primitives of previous work that serve as the basis for our algorithms. In \Cref{sec:connectivity}, we develop distributed algorithms for different shape reduction problems under the \emph{connectivity} model. In \Cref{sec:adjacency}, under the \emph{adjacency} model, we introduce an algorithm, which is based on a distributed simulation of a centralized universal transformation of~\cite{DBLP:conf/algosensors/AlmalkiGM24}. In \Cref{sec:conclusion}, we discuss potential directions for future research, including optimizing our bounds and extending our methods to other models and problems.

\subsection{Other Related Work}
\label{sec:related_work}
Algorithmic self-assembly is another area in which shape formation and related dynamics naturally arise (see, e.g., Doty~\cite{DBLP:journals/cacm/Doty12}). 
The \emph{nubot} model~\cite{DBLP:conf/innovations/WoodsCGDWY13} incorporates active molecular dynamics to self-assembly, allowing the construction of two-dimensional geometric shapes in polylogarithmic time. The authors first show how to grow fundamental structures such as chains and squares. By combining
further growth and other forms of reconfiguration, they extend these structures to obtain arbitrary shapes exponentially fast. While conceptually related to our model, a difference lies in the allowed operations, as our model only includes size-changing dynamics 

The amoebot model has been also studied as a framework for programmable matter, allowing autonomous agents to self-organize into complex configurations~\cite{DBLP:journals/dc/DaymudeRS23,DBLP:conf/spaa/DerakhshandehDGRSS14}. This model consists of simple computational units, called amoebots, that interact locally while collectively achieving global transformations. It has since been extended with more advanced communication mechanisms and structural capabilities.
One major advancement in the amoebot model is the introduction of \emph{reconfigurable circuits} by Feldmann \textit{et al.}~\cite{DBLP:journals/jcb/FeldmannPSD22}, which enable efficient communication and synchronization among nodes. These circuits allow for the instantaneous transmission of signals across dynamically formed connections.
Recent work on shortest path computation within the amoebot model has demonstrated the efficiency of reconfigurable circuits in facilitating rapid communication and coordination~\cite{DBLP:conf/podc/PadalkinS24}.  
Graphical Reconfigurable Circuits (GRCs) generalize the amoebot communication model to arbitrary graph topologies, enabling scalable and efficient distributed computations~\cite{DBLP:conf/wdag/EmekGH24}. Additionally, recent studies have extended the amoebot model to support parallel reconfiguration~\cite{DBLP:conf/sand/Padalkin0S24} which aligns with our study of growing shape and shrinking as a method for rapid programmable matter deployment before refining configurations via more complex operations. 

The problem of growing and transforming graphs has also been explored in different computational models. For example,~\cite{DBLP:journals/jcss/MertziosMSST25} studies fundamental growth processes in abstract graphs. Also,~\cite{DBLP:conf/algosensors/GuptaKMP24}  investigates collision detection in modular robots, examining the difficulty of predicting and preventing collisions under various growth and contraction operations. Related work~\cite{DBLP:journals/tcs/AlmalkiM24} explores constrained forms of expansion where growth operations are restricted to specific parts of the structure, ensuring that no collisions occur by design.

\section{Models and Problems}\label{sec:model}
\tcbset{colback=red!5!white, colframe=red!75!black, fonttitle=\bfseries}

We study a system composed of computational agents (we refer to the agents as \emph{nodes} for the rest of the paper) that initially form a connected shape $S=(V, E)$ of size $|V|= n$ on a two-dimensional square grid. Each node $u \in V$ occupies a distinct point $(u_x,u_y)$ of the grid. We adopt the model of~\cite{DBLP:conf/algosensors/AlmalkiGM24}, and we consider two graph models for connectivity among nodes. The edge set $E$ is taken as a subset of the potential edges between adjacent nodes, that is, $E \subseteq E'$ where $E'=\{uv\;|\; u,v\in V_r \text{ and } u,v \text{ are adjacent}\}$. Two nodes $u=(u_x, u_y)$ and $v=(v_x, v_y)$ are \emph{adjacent} if $u_x\in\{v_x-1,v_x+1\}$ and $u_y=v_y$ or $u_y\in\{v_y-1,v_y+1\}$ and $u_x=v_x$, that is, their orthogonal distance on the grid is one. In the connectivity model, $E$ is the set of edges in the shape $S$, while in the adjacency model, an additional update is performed at the end of every round: the graph is replaced by its adjacency closure defined formally as $AC(S)=(V, E')$, so that every pair of nodes that are adjacent on the grid is connected. We illustrate nodes as maximal circles inscribed within grid cells to represent their positions; however, our results hold for any geometry of individual nodes that does not trivially make nearby nodes intersect. In other words, while our model allows for an arbitrary geometry of individual nodes, we require that adjacent nodes are arranged so that their boundaries do not intersect.
A shape is \emph{connected} if the graph that defines it is a connected graph. Throughout the paper, we restrict attention to connected shapes.

Nodes operate as anonymous (and randomized) finite-state automata,
executing actions in discrete synchronous rounds. Each node has a compass orientation and a chirality.
Initially, the nodes may not agree on a common compass orientation or chirality.

The nodes are able to communicate through \emph{reconfigurable circuits} as follows \cite{DBLP:journals/jcb/FeldmannPSD22}.
Each edge in $E$ consists of a constant number $c$ of \emph{links}.
The constant $c$ is equal for all links.
In this paper, we will assume $c = 2$ which is sufficient and necessary for all results.
Each incident node owns one endpoint of each link which we call \emph{pins}.
Each node labels its pins of each incident edge from $1$ to $c$.
If the two nodes incident to the same edge share a chirality, then the pin with label $i$ is connected to a pin with label $c-i+1$.
Otherwise, the pin labels are matching, i.e., each pin with label $i$ is connected to a pin with label $i$. Each node partitions its pins to disjoint \emph{partition sets}.
This partitioning is called the \emph{circuit configuration} of the node.
Let $P$ denote the set of all partition sets of all nodes.
Consider the graph $C = (P, L)$, where two partition sets in $P$ are connected if they share pins of the same link.
We call the connected components of $C$ \emph{circuits}. In each round, each node can reconfigure its circuit configuration, i.e., it can change its partitioning, and send a simple signal (a \emph{beep}) on an arbitrary number of its partition sets.
This beep is received by all nodes connected to the same circuit at the beginning of the next round.
However, they do not obtain any information about the origin of the beeps or the number of beeps in case multiple nodes decide to beep on the same circuit.

After describing the static structure of the system and the communication process, we now define the operations that nodes can perform. In each round, nodes process incoming signals, transition to new states accordingly, and may execute one of the two operations---growth or shrinking---or take no action and remain in their current state.
A \emph{growth} operation applied on a node $u$ generates a new node in one of the points adjacent to $u$ and possibly translates some part of the shape. A \emph{shrinking} operation, on the other hand, is applied on a node $u$ to absorb an adjacent node $v$, and possibly translate $v$'s neighbors towards $u$. In this context, translation is the outcome of the applied operation, where the addition of new nodes or removal of existing nodes within the shape pushes or pulls its parts, causing a shift in position. One or more growth or shrinking operations applied in parallel to nodes of a shape $S$ either cause a \emph{collision} or yield a new shape $S'$ (see~\Cref{fig:operations}). There are two types of collisions: \emph{node collisions} and \emph{cycle collisions}. When describing the outcome of these operations, we assume there is an \emph{anchor} node which is stationary, and other nodes move relative to it. This is without loss of generality (w.l.o.g.) under the assumption that the constructed shapes are equivalent up to translations.

 \begin{figure}[tbp]
    \centering
    \includegraphics[scale=1.5]{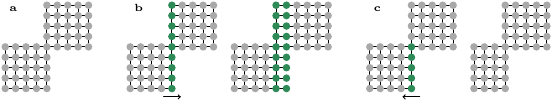}
    \caption{(a) An initial shape $S_I$. (b) Left: green nodes apply a growth operation each adding a new adjacent node (a copy of the green node). Right: the result after the growth translates the second part of the shape one unit to the east. (c) Left: green nodes apply a shrinking operation, each being absorbed and removed by their adjacent node in the shrinking direction. Right: the result after the shrinking translates the second part one unit to the west.}
    \label{fig:operations}
\end{figure}

Temporarily disregarding collisions, we begin by defining the effect of a single operation on a tree shape---first for growth, then for shrinking. Collisions will be handled later for the general case, where one or more operations are applied in parallel. Let $T=(V,E)$ be a tree and $u_0\in V$ its anchor. We set $u_0$ to be the root of $T$. A single growth operation is applied on a node $u\in V$ toward a point $(x,y)$ adjacent to $u$. 
There are two cases, which are free from collisions: ($i$) there is no edge between $u$ and $(x,y)$, ($ii$) $(x,y)$ is occupied by a node $v$ and $uv\in E$. We first define the effect in each case. In case~($i$), the growth operation generates a node $u'$ at point $(x,y)$ and connects it to $u$. In case~($ii$), assume w.l.o.g. that $u$ is closer to $u_0$ in $T$ than $v$.
Let $T(v)$ denote the subtree of $T$ rooted at node $v$. Then, the operation generates a node $u'$ between $u$ and $v$, connected to both, which translates $T(v)$ by one unit away from $u$ along the axis parallel to~$uv$. After this, $u'$ occupies $(x,y)$ and $uv$ has been replaced by $\{uu',u'v\}$. 
We refer to a \emph{neighbor handover} as the process by which a node $u$ transfers a neighbor $w$, located perpendicular to the direction of an operation (e.g., growth), to another node $u'$.
This happens by a unit translation of $T(w)$ or $T(u)$ along the axis parallel to~$uu'$, depending on which of $u,w$, respectively, is closer to $u_0$ in $T$.

We now define the effect of a single shrinking operation on a tree shape. A single \emph{shrinking} operation is applied on a node $u$ 
to absorb an adjacent node $v$, located either horizontally or vertically next to $u$. Assume that $u$ is at $(u_x,u_y)$ and $v$ at $(v_x,v_y)$. 
Let $\langle d_x,d_y\rangle= \langle x_u-x_v ,y_u-y_v\rangle \in \{ \langle \pm1, 0 \rangle, \langle 0, \pm1 \rangle \}$ denote the unit vector from $v$ toward $u$.
There are two cases: ($i$) if $v$ has no neighbors (other than $u$), the operation removes $v$ and deletes the edge $uv$; ($ii$) if $v$ has neighbors in the direction orthogonal to or opposite of $\langle d_x,d_y\rangle$, each of these neighbors is translated by one unit in the direction $\langle d_x,d_y\rangle$ pulling them towards $u$. For example, if $v$ is the east or west of $u$ in a horizontal shrinking (i.e., $d_y=0$) where a node at $(v_x,v_{y+1})$ is moved to $(v_x+d_x,v_{y+1})$, one at $(v_{x+1},v_y)$ is moved to $(v_{x+1}+d_x,v_y)$, and one at $(v_x,v_{y-1})$ is moved to $(v_x+d_x,v_{y-1})$. These translated nodes are then connected to $u$, replacing their previous connection to $v$. 
The shrinking operation is applied only when all target positions for the translated nodes are unoccupied, ensuring the process is collision-free.

Let $Q$ be a set of operations to be applied \emph{in parallel} to a connected shape $S$, each operation on a distinct pair of nodes (for shrinking and some growth cases) or on a node and an unoccupied point (only for growth). 
We assume that all operations in such a set are applied \emph{concurrently}, have the same \emph{constant execution speed}, and their \emph{duration} is equal to one \emph{round}. Let again $T=(V,E)$ be a tree, $u_0\in V$ its anchor, and set $u_0$ to be the root of $T$. We want to determine the displacement of every $v\in V\setminus\{u_0\}$, as well as any newly generated node from growth, resulting from the parallel application of operations in $Q$. As $u_0$ is stationary, and each operation either pushes or pulls part of the shape—translating nodes away from or toward the acting node depending on whether the operation is a growth or shrinking—the displacement of each $v$ is determined only by the operations along the unique $u_0v$ path in the tree.
In particular, any such operation contributes one of the unit vectors $\langle \pm1, 0 \rangle$ or $\langle 0, \pm1 \rangle$ to the motion vector $\vec{v}$ of~$v$. We can use the set of motion vectors to determine whether the trajectories of any two nodes will collide at any point. This type of collision is called a \emph{node collision} (see~\Cref{fig:node-collision}).

\begin{figure}[htbp!]
    \centering
    \includegraphics[scale=.6]{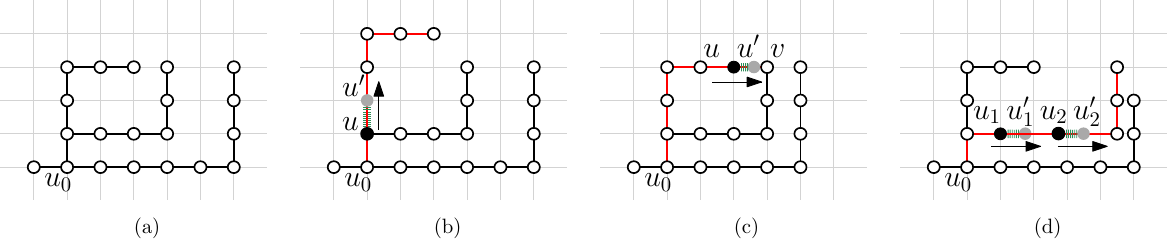}
    \caption{The figure, originally from~\cite{DBLP:journals/tcs/AlmalkiM24}, illustrates examples of node collisions. (a) Initial tree $T$. (b) The tree $T'$ after a growth operation on node $u$ moves north to generate a new node $u'$ without any collision. (c) Illustration of a node collision scenario: $T'$ here is a result of a growth operation applied on node $u$ toward the east, where a node $v$ already exists and $uv\notin E$. The newly generated node $u'$ occupies the same position as $v$, leading to a node collision. (d) Another scenario of a node collision: two nodes $u_1$ and $u_2$ simultaneously grow in the east direction and generate $u_1'$ and $u_2'$, though $u_1'$ and $u_2'$ do not collide directly, their growth pushes their branch into an adjacent branch, leading to a collision. These scenarios are analogous in the context of shrinking operations.}
    \label{fig:node-collision}
\end{figure}

Let $S$ be any connected shape with at least one cycle and any node $u_0$ be its anchor. Then, 
a set of parallel operations $Q$ on $S$ either causes a \emph{cycle collision} or its effect is essentially equivalent to the application of $Q$ on any spanning tree of $S$ rooted at $u_0$. 
Let $u$, $v$ be any two nodes on a cycle. If $p_1$ and $p_2$ are the two paths between $u$ and $v$ of the cycle, then $\vec{v}_{p_1}=\vec{v}_{p_2}$ must hold:
the displacement vectors along the paths $p_1$ and $p_2$ are equal. 
Otherwise, we cannot maintain all nodes or edges of the cycle. Such a violation is called a \emph{cycle collision} and an example is shown in~\Cref{fig:cycle-collision}. 
A set of operations is said to be \emph{collision-free} if it does not cause any node or cycle collisions.
\begin{figure}[htbp!]
    \centering
    \includegraphics[scale=0.6]{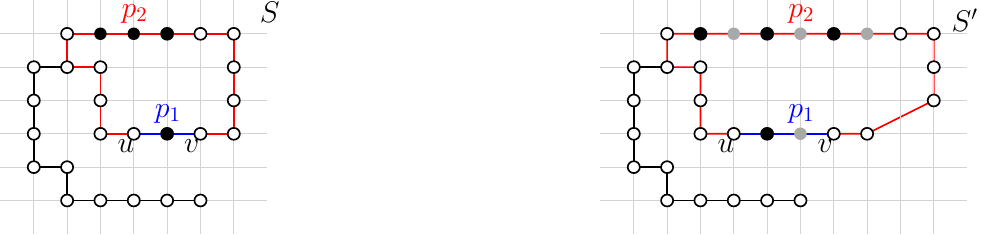}
    \caption{This figure, as presented in~\cite{DBLP:conf/algosensors/AlmalkiGM24}, shows an example of a cycle collision within the shape $S$ due to unequal displacement vectors along the two paths $p_1$ and $p_2$, thus, $\vec{v}_{p_1}\neq \vec{v}_{p_2}$. In particular, the number of generated nodes (gray nodes) along $p_2$ is greater than that along $p_1$. This difference in the number of generated nodes leads to a collision within the cycle, indicating an irregularity in the shape $S'$.
    }
    \label{fig:cycle-collision}
\end{figure}

A growth (shrinking) process $\sigma$ starts from an initial shape $S_I$ and
in each round $r\geq 1$, applies a set of parallel growth or shrinking operations---possibly a single operation---on the current shape $S_{r-1}$ to give the next shape $S_r$, until a final shape $S_F$ is reached at a round $r_f$. In the case of a growth process, we say that, \emph{$\sigma$ grows $S_F$ from $S_I$ in $r_f$ rounds}; in the case of a shrinking process, we say that \emph{$\sigma$ shrinks $S_I$ to $S_F$ in $r_f$ rounds}. We also assume that parallel operations have the same cardinal direction and that a node gets at most one operation per round. These assumptions simplify the description of algorithms and can be easily dropped.

A node $u \in V$ of a shape $S$ is a leaf if and only if its degree is one, i.e., $\delta(u) = 1$. Let $L \subseteq V$ denote the set of all leaves.
A node $u \in V$ of a shape $S$ is called a \emph{turning point}
if either $u \in L$ or if there are two neighbors $v_1$ and $v_2$ of $u$ such that $v_1u$ is perpendicular to $uv_2$. For example, in the special case of a path shape $P=( u_1, u_2, \ldots, u_n )$, a node $u_i$ is a turning point if either $i \in \{1,n\}$ or $u_{i-1}u_i$ is perpendicular to $u_iu_{i+1}$. For uniformity of our arguments, we add the endpoints of a path to the set of its turning points. We denote the set of all turning points by $TP \subseteq V$, and by definition, $L \subseteq TP$.
A node $u_j \in P$ is a segment node if it lies between any two consecutive turning points and satisfies $\delta(u) = 2$. 
The sequence of segment nodes between two turning points forms a segment $s$ of length $\ell$, where $\ell$ is the number of nodes in the segment, including both turning points. 

We define the \emph{relative position} of $v = (v_x, v_y)$ with respect to $u = (u_x, u_y)$ by $r(u,v) = (r_H(u,v), r_V(u,v)) \in \{ \rightarrow, 0, \leftarrow\} \times \{ \uparrow, 0, \downarrow\}$ such that, $r_H(u,v) = \{ \rightarrow \text{ if } u_x < v_x, \; 0 \text{ if } u_x = v_x, \; \leftarrow \text{if } u_x > v_x \}$,  and  $r_V(u,v) = \{ \uparrow \text{ if } u_y < v_y, \; 0 \text{ if } u_y = v_y, \; \downarrow \text{ if } u_y > v_y \}$. 
We call any two trees $T$ and $T^{\prime}$ \emph{geometrically equivalent} if and only if there exists a one-to-one correspondence between their turning points such that only the segment lengths and the relative positions of the turning points may differ. We call two \emph{geometrically equivalent} trees $T$ and $T^{\prime}$ additionally \emph{topologically equivalent} if and only if the relative positions of all corresponding turning points are the same. 

A \emph{column} or \emph{row of a shape} $S$ refers to a row or column of the grid that contains at least one node of $S$. We call a column or row of a shape $S$ \emph{compressible} if all nodes in that row or column are segment nodes, i.e., the row or column contains no turning points. Otherwise, it is \emph{incompressible}. We call the nodes of incompressible columns or rows \emph{incompressible nodes}. We define the incompressible form of a shape $S$, denoted by $i(S)$, as the shape obtained by removing all compressible rows and columns from $S$. The resulting shape is incompressible and therefore topologically equivalent to $S$~\cite{DBLP:conf/algosensors/AlmalkiGM24}.

\subsection{Problems}
\label{sec:problem}

We study the problem of reducing an initial shape $S_I = (V_I, E_I)$ to a target shape $S_F = (V_F, E_F)$, where $|V_F|\leq |V_I|$. We assume that the constructed shapes are equivalent up to translations. In what follows, we restrict attention to $S_I$ and $S_F$ being trees, denoted $T_I$ and $T_F$, respectively. Throughout, $n$ denotes the number of nodes and $k$ the number of turning points in $T_I$. In general, the system must reduce itself from an initial shape $S_I = (V_I, E_I)$ to a given target shape $S_F = (V_F, E_F)$, where $|V_F|\leq |V_I|$. For any pair of geometrically or topologically equivalent trees $T_I,T_F$, we assume that $s_i$ denotes the same segment with lengths $\length{s^I_i}$ and $\length{s^F_i}$ in $T_I$ and $T_F$, respectively. For such trees, the input convention assumes that in $T_I$ each segment $s^I_i$ stores a binary representation of $\length{s^I_i}$ and $\length{s^F_i}$.
We study the following problems in the connectivity model, apart from \textsc{Single Node Reduction} which we also study in the adjacency model.

\smallskip

\noindent\textsc{Single Node Reduction}.
Let $T_I$ be an initial tree.
The system must reduce itself from $T_I$ to $T_F = (\{u_0\}, \emptyset)$ while avoiding collisions. In the adjacency model, the initial shape is instead any connected shape. 

\smallskip

\noindent\textsc{Geometrical Reduction}.
As in \textsc{Single Node Reduction}, with $T_F$ now being a tree geometrically equivalent to $T_I$, satisfying $\length{s^F_i}\leq\length{s^I_i}$ for all segments $s_i$. 

\smallskip

\noindent\textsc{Topological Reduction}.
As in \textsc{Geometrical Reduction}, with $T_F$ additionally being topologically equivalent to $T_I$.

\smallskip

\noindent\textsc{Incompressible Reduction}.
A restricted version of \textsc{Topological Reduction}, where $T_F = i(T_I)$. 

\smallskip

We also consider a variant of \textsc{Single Node Reduction} and \textsc{Incompressible reduction} where each node additionally knows in advance if it is an incompressible node.

\section{Preliminaries}
\label{sec:prelim}

In this section, we highlight fundamental primitives presented in previous work, which form the basis for our proposed algorithms. Feldmann \emph{et al.} \cite{DBLP:journals/jcb/FeldmannPSD22} have proposed algorithms that allow us to elect a unique leader node, to align the compass orientation of all nodes, and to obtain an agreement about the chirality of all nodes.
Each of these algorithms requires $O(\log n)$ rounds w.h.p. Our algorithms utilize (some of) these algorithms in a preprocessing step which we will explicitly state for each algorithm.

In case we perform some subroutine with varying runtimes in different parts of the shape in parallel, we use the synchronization mechanism of \cite{DBLP:journals/nc/PadalkinSW24} to keep them synchronized. This only adds a constant factor to the runtime.

Let $(x_{b-1}, \dots, x_0)$ be the binary representation of value $x$.
We store $x$ in a segment $s = \{u_0, \dots, u_{m-1}\}$ of nodes where $m \geq b/c$ for a constant $c$ such that for each $i \in \{ 0, \dots, b-1 \}$, $x_i$ is stored by $u_{\lfloor i/c \rfloor}$.
We can \emph{transfer} the value to another segment $s' = \{v_0, \dots, v_{\ell'-1}\}$ with $\ell' \geq b/c$ in $O(b)$ rounds \cite{DBLP:journals/nc/PadalkinSW24}.
If a segment holds multiple values, it can perform basic arithmetic and other operations.
Addition, subtraction, and comparisons require $O(1)$ rounds while multiplication and division require $O(b)$ rounds where $b$ denotes the number of bits of the larger value \cite{DBLP:journals/corr/abs-2501-16892}.

Another powerful primitive is the \emph{primary and secondary circuit} (PASC) algorithm of Padalkin \emph{et al.} \cite{DBLP:journals/nc/PadalkinSW24} (a prototype was published by Feldmann \emph{et al.} \cite{DBLP:journals/jcb/FeldmannPSD22}).
Let $s = \{u_0, u_1, \ldots, u_{m-1}\}$ be a segment of $m$ nodes.
We assume that the segment is ordered (e.g., by a cardinal direction $\vec d \in \{ \rightarrow, \uparrow, \leftarrow, \downarrow \}$) and that each node knows its predecessor and successor with respect to that order (if they exist, respectively).
Let $u_r$ be an arbitrary node of $s$.
We call $u_r$ the \emph{reference node}.
The PASC algorithm computes $j - r$ for each node $u_j$.
Since in general, a node $u_j$ does not have enough memory space to store $j - r$, it will compute it iteratively bit by bit.
We will describe the primitive to a detail sufficient to understand the results of this paper.

In each iteration, the nodes establish two circuits along the segment as follows.
Each node utilizes two partition sets: a \emph{primary partition set} and a \emph{secondary partition set}.
Each node is either in an \emph{active} or \emph{passive} state.
Initially, each node is active.
If a node is active, it connects its primary (secondary) partition set to the secondary (primary) partition set of its predecessor.
If a node is passive, it connects its primary (secondary) partition set to the primary (secondary) partition set of its predecessor.
We obtain two circuits along the segment (see \Cref{fig:pasc}). Now, the reference node $u_r$ beeps on its primary partition set.
Each node receives that beep either on its primary or secondary partition set.
Each active node that receives a beep on its secondary partition set becomes passive in the next iteration.
We terminate if there is no such node.

By interpreting a beep on the primary partition set as a $0$ and a beep on the secondary partition set as a $1$, each node $u_j$ computes the two's complement representation of $j - r$.
More precisely, in the $i$-th iteration of the PASC algorithm, each node $u_j$ computes the $i$-th bit of $j - r$.
We refer to \cite{DBLP:journals/nc/PadalkinSW24} for the proof and more details.
Note that the last node $u_{m-1}$ computes $m - 1$ if we choose the first node $u_0$ as the reference node $u_r$.
Further, each node $u_j$ can perform arithmetic operations and comparisons with values of constant size or with other values if we broadcast them at the same time. We obtain the following lemma.

\begin{lemma}[Feldmann \emph{et al.} \cite{DBLP:journals/jcb/FeldmannPSD22}, Padalkin \emph{et al.} \cite{DBLP:journals/nc/PadalkinSW24}]
\label{lem:pasc:segment}
    Let $s$ be a segment of $m$ nodes.
    The PASC algorithm computes the length of $s$ and stores it within $s$ in $O(\log m)$ rounds.
    Given an $m' \leq m$, the PASC algorithm also identifies the first $m'$ nodes of $s$ in $O(\log m)$ rounds.
\end{lemma}

Let $u = (u_x, u_y)$ and $v = (v_x, v_y)$ be two nodes.
Let $|u_x - v_x|$ ($|u_y - v_y|$) denote the \emph{horizontal (vertical) distance} of $u$ to $v$.
The PASC algorithm can be adapted to compute the horizontal (vertical) distance of each node $u_j$ to a reference node $u_r$.
We refer to \cite{DBLP:journals/nc/PadalkinSW24} for more details.

\begin{lemma}[Padalkin \emph{et al.} \cite{DBLP:journals/nc/PadalkinSW24}]
\label{lem:pasc:spatial}
    Let $S$ be a shape of $m$ nodes and $u_r$ an arbitrary \emph{reference node} of $S$.
    The spatial PASC algorithm computes the sign of the horizontal (vertical) distance of $u_r$ to each node $u_j$ and stores it in $u_j$ in $O(\log m)$ rounds.
\end{lemma}

\begin{figure}[htbp!]
    \centering
    \includegraphics[scale=1]{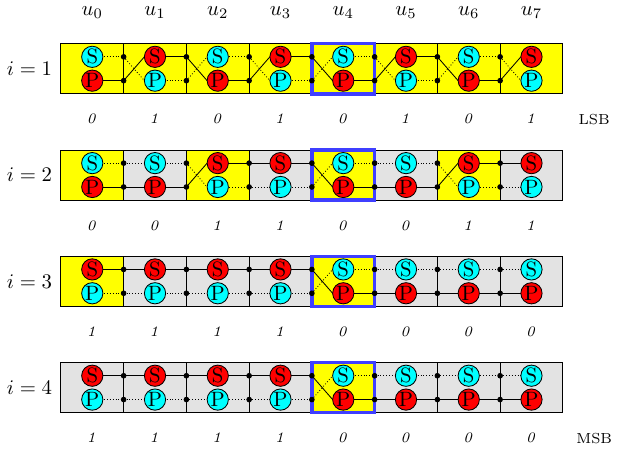}
    \caption{ Example of an execution of the PASC algorithm.
    Note that the nodes are depicted as rectangles and the external links are reduced to a single nodes.
    The reference node $u_r = u_4$ is marked by a blue border.
    The yellow nodes are active while the gray nodes are passive.
    The partition sets marked by a P (S) are the primary (secondary) partition sets.
    The nodes receive the beep on the red partition sets.
    Figure adapted from \cite{DBLP:journals/nc/PadalkinSW24}.
    }
    \label{fig:pasc}
\end{figure}

\section{Connectivity Model}\label{sec:connectivity}

We start our examination in the connectivity model.
In each of the following subsections, we study one of the four reduction problems: \textsc{Single Node Reduction}, \textsc{Geometrical Reduction}, \textsc{Incompressible Reduction}, and \textsc{Topological Reduction}.

\subsection{Single Node Reduction
}\label{subsec:tree:single:node}
In this section, we present the \emph{BFS shrinking} algorithm, which reduces any arbitrary tree $T$ to a single node. We assume common chirality or establish it in a preprocessing phase (see \Cref{sec:prelim}). Note that we do not apply the synchronization mechanism from~\cite{DBLP:journals/nc/PadalkinSW24}, as operations progress uniformly across the shape. The following presents this in detail.

\medskip
\noindent\textbf{BFS shrinking.} This algorithm consists of three subroutines, \emph{segment detection}, \emph{segment coloring and shrinking}, and \emph{final segment shrinking}. 
First, we identify segments starting from the leaves and progressing toward their turning points. Once a segment is identified, it proceeds to \emph{segment coloring and shrinking} subroutine, where it is colored and shrunk by parity-based operations. 
This recursive process continues until $T$
reduces to a single node or segment. In the latter case, the \emph{final segment shrinking} subroutine handles simultaneous leaf formation at both endpoints. The algorithm completes in
$O(k\log n)$ rounds.

\medskip
\noindent\emph{Segment detection.} 
This subroutine allows nodes to detect leaves in their segment and determine their direction by establishing two parallel circuits along each segment $s_i$. For simplicity, we assume all nodes share the same chirality (clockwise or counterclockwise). Otherwise, we establish this assumption in a preprocessing phase (see \Cref{sec:prelim}).

Each node $u \in s_i$ has two disjoint partition sets $P_1(u)$ and $P_2(u)$ where $P_1(u)$ connects the first pin on one side to the second pin on the opposite side, and $P_2(u)$ connects the second pin on the former side and to the first pin of the latter side. This results in two parallel circuits along each segment. It is important to note that each segment's circuit operates independently. The leaf node sends a signal (beep) on its second pin (see \Cref{fig:segment_detection}), this beep propagates through the segment, and each node $u\in s_i$ receives the signal on one of its partition sets. If the signal is received on the first pin of a partition set $P_i(u)$, that pin points toward the leaf. If the signal is received on the second pin $P_i(u)$, that pin points to the opposite end of the segment.

Now, there are three possible scenarios for each segment $s_i$.
First, if there is no signal (beep) in the segment $s_i$, neither endpoints of $s_i$ is a leaf. In this case, the segment remains inactive, waiting for one of its endpoints to become a leaf. Second, if there is a beep on only one circuit, exactly one endpoint is a leaf. The segment then proceeds to the \emph{segment coloring and shrinking} subroutine. 
Third, if there is a beep on both circuits (i.e., two leaves send signals simultaneously), both endpoints of the segment are leaves. In this case, we proceed with the \emph{final segment shrinking } subroutine where the leaves perform a leader election.

\begin{lemma}\label{lem:tree-orientation}
    Each segment agrees on a common orientation once one of its endpoints is a leaf. This process requires $O(1)$ rounds for all segments.
\end{lemma}

\begin{proof}
In the proof we consider two cases. In the first case, a segment $s_i$ has distinct endpoints: one endpoint is a turning point and the other is a leaf. In the second case, two segments share a common turning point while each has its own leaf at the opposite end.
Consider the first case where a segment $s_i$ terminates at a turning point. Let $s_i$ be a segment consists of a sequence of nodes $\{u_0, u_2, \ldots, u_{l_{i-1}}\}$, where $u_0$ and $u_{l_{i-1}}$ are the endpoints of $s_i$, specifically, $u_0$ represents the turning point ($tp$), and $u_{l_{i-1}}$ represents the leaf.
Each node in $s_i$ has two parallel circuits defined by the partition sets $P_1$ and $P_2$. If a signal propagates along $s_i$, any intermediate node $u_j\in s_i$ where $1 \leq j \leq \ell_i - 2$ receives the signal on one of its partition sets, if it is on the first $P_1(u_j)$, the corresponding pin points toward the leaf $u_{l_{i-1}}$ of $s_i$. If it is on the other partition set $P_2(u_j)$, the corresponding pin points toward the other endpoint $u_0$ ($tp$) of $s_i$. 

Now consider the second case where a turning point $tp$ is shared by two horizontal segments $s_1 = \{u_0, u_1, \ldots, u_{l_{i-1}}\}$ and $s_2 = \{v_0, v_1, \ldots, v_{l_{i-1}}\}$. Here, $tp$ corresponds to the endpoint $u_0$ in $s_1$ and $v_0$ in $s_2$. Both leaves $u_{l_{i-1}}$ and $v_{l_{i-1}}$ initiate signals that propagate unidirectionally toward $tp$ ($u_0$) in $s_1$, and ($v_0$) in $s_2$. For any intermediate node $u_j, \in s_1,$ and $v_j \in s_2$, where $1 \leq j \leq \ell_i - 2$, let us assume that $u_j$ receives the signal on the first $P_1(u_j)$, that pin points towards the leaf $u_{l_{i-1}}$ of $s_1$, and $v_j$ receives the signal on the second $P_2(v_j)$, that pin points towards the end $tp$ ($v_0$) of the segment $s_2$. In other words, one segment will connect to $P_1$ at $tp$, while the other connects to $P_2$ at $tp$ on the opposite side, as shown in Figure~\ref{fig:segment_detection}. Regardless of the case (horizontal segments, vertical segments, or mixed horizontal and vertical), each signal propagates independently along its segment and terminates at $tp$. The partition sets $P_1$ and $P_2$ ensure that each signal connects to a distinct partition set at $tp$, which defines the orientation of the signal in the tree. 
As each signal---initiated by a leaf node---propagates independently and is processed locally at each node along the segment, and since this processing occurs only once for each signal, the entire procedure completes in $O(1)$ rounds for all segments.
\end{proof}

\begin{figure}[tbp]
    \centering
    \includegraphics[scale=0.7]{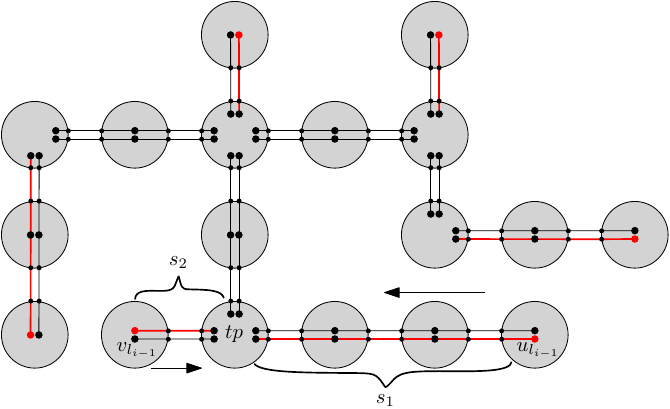}
     \caption{An illustration of the segment detection subroutine which shows the parallel circuits configuration. Each leaf (e.g.,$u_{l_{i-1}},v_{l_{i-1}}$) initiates a signal on its second pin, which propagates along the red highlighted circuit of segments $s_1$ and $s_2$. At the turning point $tp$ both $s_1$ and $s_2$ segments converge and by the partition sets each signal is processed independently.}
    \label{fig:segment_detection}
\end{figure}

After applying the \emph{segment detection} subroutine, we proceed with the following process.
\noindent\emph{Segment coloring and shrinking}. In order to compute a coloring we will make use of the PASC algorithm by Padalkin \emph{et al.}~\cite{DBLP:journals/nc/PadalkinSW24}. In the first round, the algorithm computes the parity of the distance of each node from a reference node. Let $\mathcal{S}_{tp}$ denote the set of segments converging at the turning point $tp$. For each segment $s_i \in \mathcal{S}_{tp}$, where $s_i=\{u_0,u_1,\ldots, u_{\ell_{i-1}}\}$ is a segment with $u_0=tp$ and $u_{\ell_{i-1}}$ a leaf, we apply the first round of the PASC algorithm. Each node $u_j\in s_i$  has two partition sets: primary and secondary. We use the same configurations defined in~\cite{DBLP:journals/nc/PadalkinSW24} to connect these partition sets, specifically, the primary partition set of $u_j$ is connected to the secondary partition set of its predecessor $u_{j-1}$. The secondary partition set of $u_j$ is connected to the primary partition set of its predecessor $u_{j-1}$. 
This configuration forms two disjoint circuits for each segment $s_i$, in these circuits, the partition sets alternate between primary and secondary. We then use $tp$ as a reference node for $s_i$ to compute the parity of the distance of each node $u_j \in s_i$ from $tp$.  Once the parity is computed, we color all even-parity nodes blue and all odd-parity nodes green (see Figure~\ref{fig:segment_coloring}). This assignment, using the first round of the PASC algorithm, is completed in a single round. 

\begin{figure}[tbp]
    \centering
    \includegraphics[scale=0.8]{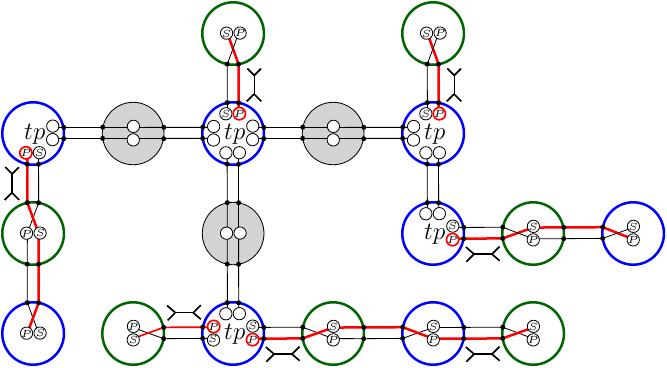}
    \caption{An illustration of PASC applied to each segment to compute parity with respect to the turning points.}
    \label{fig:segment_coloring}
\end{figure}

After coloring segments using PASC, we do the shrinking process. For every node $u_j$ with even-parity (blue) in a segment $s_i$, let $u_{j+1}$ be its successor with odd-parity (green). The even-parity node $u_j$ absorbs and pulls its edge $(u_j,u_{j+1})$ towards $u_{j}$, such that the primary circuit of $u_{j+1}$'s successor (i.e., $u_{j+2}$) connects to the secondary circuit of $u_{j}$, and the secondary circuit of $u_{j+1}$'s successor (i.e., $u_{j+2}$) connects to the primary circuit of $u_{j}$. After shrinking all odd-parity nodes, the PASC algorithm is recomputed on $s_i$, updating the parity and repeating the shrinking process. 

Let $tp_i$ be a turning point in $T$ connected to segment set $\mathcal{S}_{tp_i} = \{s_1, s_2, \ldots, s_{m_i}\}$, where $m_i$ is the number of segments connected to $tp_i$. Although the shrinking operations on the segments can be performed in parallel, the nodes immediately adjacent to $tp_i$ synchronize so that $tp_i$ does not update its state until all connected segments have finished their shrinking process. Specifically, if multiple signals are received at $tp_i$, it prioritizes the first segment from which it receives the signal and processes it accordingly. The shrinking process then continues for each segment, where each node in $s_j\in \mathcal{S}_{tp_i}$ absorbs its neighbor according to the computed PASC. The turning point $tp_i$ becomes a leaf after all segments in $\mathcal{S}_{tp_i}$ have shrunk completely. Each segment $s_j$ has length $\ell_j$, and the shrinking process of $s_j$ halves its length in each round, completing in $O(\log \ell_j)$ rounds. Once all segments in $\mathcal{S}_{tp_i}$ have been shrunk completely, the turning point $tp_i$ becomes a leaf, and the process then proceeds to the next iteration.

\begin{lemma}\label{lem:pasc-on-segment}
    For any segment $s_i$, each iteration of the \emph{segment coloring} subroutine completes in $O(1)$ rounds.  
\end{lemma}
\begin{proof}
For any segment $s_i = \{u_0, u_2, \ldots, u_{\ell_{i-1}}\}$ in a tree $T$, where $u_0 = tp$ is a turning point and $u_{\ell_{i-1}}$ is a leaf.
The turning point $tp$ initiates a signal that propagates through the segment $s_i$. Each node $u_j \in s_i$ receives the signal through either its primary or secondary circuit, which computes its parity. Nodes receiving the signal on their primary circuit are assigned even parity, while those receiving it on their secondary circuit are assigned odd parity. The circuit configuration enforces this alternating assignment, as proved in~\cite{DBLP:journals/nc/PadalkinSW24}. Since the process involves only a single signal initiation from $tp$, thus, the coloring process completes in $O(1)$ rounds. 
\end{proof}

\begin{lemma}\label{lemma:parallel-shrinking}
    The \emph{segment coloring and shrinking} subroutine completes in $O(\log \ell)$ rounds, where $\ell$ is the length of the segment. 
\end{lemma}

\begin{proof}
    Assume a segment $s$ with a length $\ell$. After coloring the segment nodes using PASC into blue (even-parity) and green (odd-parity) sets, the odd-parity nodes are shrunk. This operation reduces the length of the segment by approximately half in each round, though the shrinking is not always exact due to specific segments length. Specifically, if $\ell$ is even or odd the segment length is reduced by $\lfloor\ell/2\rfloor$.
    Let the segment length after $r$ rounds be $\ell_r = \lfloor{\ell}/{2^r}\rfloor$. The shrinking process continues until $\ell_r = 1$, which implies ${\ell}/{2^r} = 1$, thus, $r = \lceil\log \ell\rceil$. Since each round involves constant-time operations for coloring and the segment's length is halved in each round, the time complexity of the shrinking process of a segment $s$ with length $\ell$ is $O(\log \ell)$.
\end{proof}

\noindent\emph{Final segment shrinking.} When $T$ is shrunk to a single segment,  it is not possible to determine an orientation during the segment detection subroutine, which implies that we cannot decide which endpoint of the segment will perform the PASC for the coloring and shrinking process. 
We perform a leader election between the leaves, we reuse the established circuits from the \emph{segment detection} subroutine (i.e., two parallel circuits along the segment).
Both leaves toss a coin and beep on their circuit if they toss heads. If there is a beep on only one circuit, i.e., only one leaf has tossed heads, that leaf is elected as the leader, while the other leaf acts as a turning point. All segment nodes determine the leader's direction as in the segment detection subroutine, after which we proceed to the segment coloring and shrinking subroutine. However, if both leaves beep (i.e., indicating a tie), the procedure repeats.

\begin{lemma}\label{lem:final-leader}
    The final segment agrees on a common orientation through a leader election process, which requires $O(\log n)$ rounds w.h.p.
\end{lemma}
\begin{proof}
    To elect a leader for the final segment of length $\ell$, we perform a leader election process using the established circuits as described above. Both leaves of the segment toss a coin independently in each iteration. A leader is elected if exactly one of the leaves tosses \emph{heads} and beeps on its circuit, while the other tosses \emph{tails}.  The process fails in an iteration if both leaves toss the same result (either both heads or both tails).
    The probability of failure in a single iteration is $\Pr[\text{no leader}]=\frac{1}{2}\cdot\frac{1}{2}+\frac{1}{2}\cdot\frac{1}{2}=\frac{1}{2}$, and the success probability in one iteration is $\Pr[\text{leader}]=1-\Pr[\text{no leader}]=\frac{1}{2}$. After $k$ iterations, the probability that no leader has been elected is ${(\frac{1}{2})}^k$. To ensure the process succeeds w.h.p.\ we bounded the failure probability with $\frac{1}{n^c}$. So, ${(\frac{1}{2})}^k \leq \frac{1}{n^c}$ and therefore $k \geq c \log n$. Thus, to ensure success w.h.p., the process requires at least $c \log n$. Each iteration requires a single round, so the total worst-case runtime for the process to succeed w.h.p.\ is $O(\log n)$ rounds. 
\end{proof}

By Lemmas~\ref{lem:tree-orientation}, \ref{lem:pasc-on-segment}, \ref{lemma:parallel-shrinking}, \ref{lem:final-leader}, we conclude the following:
\begin{theorem}\label{theo:shrinking-tree}
   After an $O(\log n)$-round preprocessing (w.h.p.), the \emph{BFS shrinking} shrinks any tree $T$ to a single node in $O(k \log n)$ rounds w.h.p.
\end{theorem}
\begin{proof}
Let $T$ be tree with $n$ nodes and $k$ turning points. The tree $T$ consists of turning points $\{tp_1, tp_2, \ldots, tp_k\}$, where each turning point $tp_i$ is connected to a set of segments $\mathcal{S}_{tp_i} = \{s_1, s_2, \ldots, s_{m_i}\}$, with $m_i = \delta(tp_i) - 1$. Since we are on two dimensional square grid $\delta(tp_i)$ is bounded by $m_i\leq 4$. 
By using \emph{segment detection, coloring and shrinking} subroutines (Lemmas~\ref{lem:tree-orientation}, \ref{lem:pasc-on-segment}, and \ref{lemma:parallel-shrinking}), each segment $s_i$ is processed independently, shrinking its length in $O(\log \ell)$ rounds, where $\ell$ is the length of the segment.
The \emph{BFS shrinking} algorithm progresses recursively, shrinking all segments in $\mathcal{S}_{tp_i}$ connected to a turning point $tp_i$. Once all segments in $\mathcal{S}_{tp_i}$ are fully shrunk, $tp_i$ becomes a leaf of its parent turning point $tp_p$, which means $tp_i$ is part of one of the segments in $\mathcal{S}_{tp_p} = \{s_1, s_2, \ldots, s_{m_p}\}$, such that $s_j = \{tp_p=u_0, u_2, \ldots, u_{\ell_{j-1}}=tp_i\}$. This recursive process continues until $T$ is reduced to either a single node or a single segment. If $T$ is reduced to a single node, it terminates, if it is reduced to a single segment, the \emph{final segment shrinking} subroutine (Lemma~\ref{lem:final-leader}) is performed, which determines the segment orientation in $O(\log n)$ rounds w.h.p. 
Since the recursion processes $k$ turning points and each step includes shrinking segments of length of at most $n$, each step in the recursion takes $O(\log n)$ rounds. Thus, the total time complexity of this approach is $O(k\log n)$ rounds.
\end{proof}

\subsection{Geometrical Reduction}\label{subsec:geo:eqv:tree}
In this section, we consider the reduction of a tree $T_I$ into a geometrically equivalent tree $T_F$.
We will show some lower bounds for the case that we are only allowed to perform shrinking operations.
Note that in general, it is not always possible to perform such a reduction with only shrinking operations.

\begin{theorem}
\label{the:lowerBound}
    Let $\cal A$ be an algorithm that reduces a path $T_I$ to a geometrically equivalent path $T_F$ using only shrinking operations.
    Then, there exists an infinite family of pairs of paths of $k = \Theta(\sqrt[3]{n})$ turning points each, for which $\cal A$ requires $\Omega(k \log k)$ rounds.
\end{theorem}

\begin{proof}
    We prove the theorem by giving a pair of paths for which any shrinking algorithm will take $\Omega(k\log n)$ rounds.
    Let $k \geq 4$.
    Our initial path $T_I$ consists of two spirals $P_b$ (colored black) and $P_r$ (colored red) as shown in Figure~\ref{fig:lowerBound} (a).
    Our final path $T_F$ consists of a subpath $P_b$ and a subpath $P_g$ (colored green) that is geometrically equivalent to $P_r$, as shown in Figure~\ref{fig:lowerBound} (b).
    We now describe the construction of $P_b$, $P_r$ and $P_g$.

\begin{figure}[htbp]
    \centering
    \includegraphics[scale=0.9]{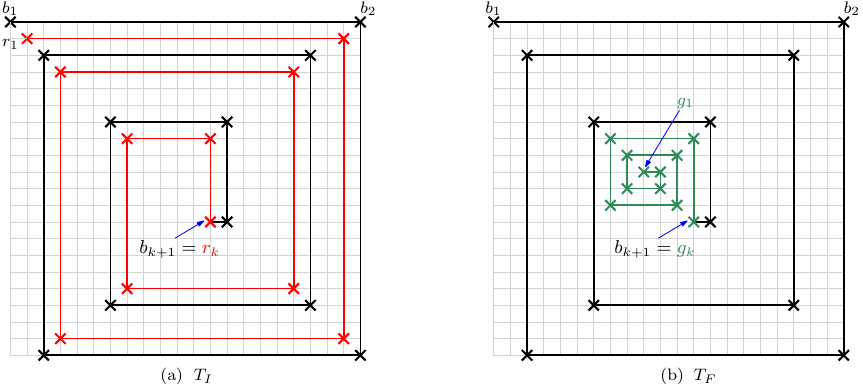}
    \caption{
        (a) An initial path $T_I$. (b) A final path $T_F$.
        Each algorithm $\cal A$ has to reduce the subpath between $r_1$ and $r_i$ to the subpath between $g_1$ and $g_i$ before it can start to reduce $r_i r_{i+1}$ to $g_i g_{i+1}$.
    }
    \label{fig:lowerBound}
\end{figure}

    Let the turning points of $P_b$, $P_r$ and $P_g$ be labeled as $(b_1, b_2, \ldots, b_k, b_{k+1})$, $(r_1, r_2, \ldots, r_k)$ and $(g_1, g_2, \ldots, g_k)$ respectively as shown in Figure~\ref{fig:lowerBound}.
    Note that $r_i = g_i$ for all $i \in [1,k-1]$.
    Let $P$ be a path and $u, v \in V(P)$.
    Then, we denote the subpath of $P$ between $u$ and $v$ as $P[u,v]$.
    Moreover, we denote the length of $P$ by $|P|$.
    Then, we have the following.
    \begin{align*}
        |P_g[g_i,g_{i+1}]| &= \lceil i/2 \rceil + 1 \text{ for all } i \in [1,k-1] \\
        |P_b[b_k,b_{k+1}]| &= 2 \\
        |P_b[b_{k-1},b_{k}]| &= |P_g[g_{k-1},g_{k}]| + 1 \\
        |P_b[b_{k-2},b_{k-1}]| &= |P_g[g_{k-2},g_{k-1}]| + 1 \\
        |P_b[b_{i},b_{i+1}]| &= |P_b[b_{i+2},b_{i+3}]| + \lceil i/2 \rceil + 1 \text{ for all } i \in [1,k-3] \\
        |P_r[r_{k-1},r_{k}]| &= |P_b[b_{k-1},b_{k}]| - 1 \\
        |P_r[r_{i},r_{i+1}]| &= |P_b[b_{i},b_{i+1}]| - 2  \text{ for all } i \in [1,k-2]
    \end{align*}

    By solving these recursive equations, we get that $n = |P| = \Theta(k^3)$.
    Furthermore, $|P_r[r_i,r_{i+1}]| = \Omega(k)$ for each $i \in [1,k-1]$.
    Observe that, by construction, $T_I$ and $T_F$ are geometrically equivalent.
    We will prove Theorem~\ref{the:lowerBound} by proving the following claim.
    
    \begin{claim}\label{cl:lowerBound}
        Let $\cal A$ be an algorithm that reduces $T_I$ to $T_F$.
        Then, $\cal A$ will perform a shrinking operation on the segment $P_r[r_i, r_{i+1}]$ only when the subpath $P_r[r_1, r_i]$ is reduced to $P_g[g_1, g_i]$.
    \end{claim}
    
    \begin{claimproof}
    
        Assume by contradiction, $\cal A$ performs a shrinking operation on the segment $P_r[r_i,$ $r_{i+1}]$ before the subpath $P_r[r_1, r_i]$ is reduced to $P_g[g_1, g_i]$.
        Note that $1 < i < k-2$ holds since we do not perform any operations on $P_r[r_{k-2}, r_k]$ by construction.
        This operation will result in reducing the distance between the segments $P_r[r_{i-1}, r_i]$ and $P_b[b_{i+3}, b_{i+4}]$.
        If $i = 2$, $P_r[r_1,r_2]$ will collide with $P_b[b_1,b_2]$, which would result in a contradiction.
        So, $i > 2$ must hold.
        Due to the construction of $T_I$, the distance between $P_r[r_{i-1}, r_i]$ and $P_b[b_{i+3}, b_{i+4}]$ is equal to $|P_g[g_{i-2}, g_{i-1}]|$.
        Hence, reducing that distance prevents us from reducing $P_r[r_{i-1}, r_i]$ without causing a collision, which would result in a contradiction.
        However, another operation in a segment $P_r[r_j,$ $r_{j+1}]$ with $j > i$ may increase the distance between $P_r[r_{i-1}, r_i]$ and $P_b[b_{i+3}, b_{i+4}]$ again.
        In this case, we repeat the argument with $P_r[r_j,$ $r_{j+1}]$.
        Since the number of segments is bounded by $k$, we will eventually reach a segment that results in a contradiction.
    \end{claimproof}

    Due to Claim~\ref{cl:lowerBound}, we can conclude that any shrinking algorithm will shrink $P_r$ by shrinking the segments $P_r[r_1, r_2], P_r[r_2, r_3], \ldots, P_r[r_{k-1}, r_k]$ in this order.
    The length of each segment is $\Omega(k)$, thus each segment can be reduced to its final length in $\Omega(\log k)$ rounds.
    As there are $k$ turning points, we achieve our bound as stated in Theorem~\ref{the:lowerBound}.
\end{proof}

The following corollary shows that we cannot hope for a $o(\log^2 n)$-rounds algorithm even for small $k = \Theta(\log n)$.

\begin{corollary}\label{cor:geom-log2-lower-bound}
    Let $\cal A$ be an algorithm that reduces a path $T_I$ to a geometrically equivalent path $T_F$ using only shrinking operations.
    Then, there exists an infinite family of pairs of paths of $k = \Theta(\log n)$ turning points each, for which $\cal A$ requires $\Omega(\log^2 n)$ rounds.
\end{corollary}

\begin{proof}
\label{cor:lowerBound}
    We adapt our construction of \Cref{the:lowerBound} as follows.
    If $k$ is even (odd), we add $\exp k / k$ nodes to each $P_b[b_i,b_{i+1}]$ and $P_r[r_i,r_{i+1}]$ with an odd (even) $i$.
    Overall, we change the length of $2 \lfloor k/2 \rfloor$ many segments.
    By solving the adapted recursive equations, we get $n = |P| = \Theta(\exp k)$.
    Moreover, $|P_r[r_i,r_{i+1}]| = \Omega(\exp k / k)$ for half of the $i \in [1,k-1]$.
    Observe that, by construction, $T_I$ and $T_F$ are geometrically equivalent.
    The corollary follows from the fact that \Cref{cl:lowerBound} still holds.
\end{proof}

\subsection{Incompressible Reduction}\label{subsec:incomp:tree}

In this section, we present the \emph{incompressible tree} algorithm which reduces an initial arbitrary tree $T_I$ 
to its incompressible form $i(T_I)$.
We assume that we have a leader node given and that the nodes share a common compass orientation and chirality.
Otherwise, we establish these assumptions in a preprocessing phase (see \Cref{sec:prelim}).
This allows us to define a direction for each segment (w.l.o.g. $\rightarrow$ and $\downarrow$) and with that an order.

\noindent\textbf{Incompressible tree reduction}.
In the following, we present the \emph{incompressible tree} algorithm.
For now, we assume that each node knows whether it is an incompressible node and with that whether it belongs to a compressible column and row.
Almalki \emph{et al.} \cite{DBLP:conf/algosensors/AlmalkiGM24} have proven that in order to obtain the incompressible shape, we have to compress all compressible columns and rows.

The nodes of a maximal sequence of consecutive compressible columns (rows) form sets of parallel horizontal (vertical) subsegments of the same length \cite{DBLP:conf/algosensors/AlmalkiGM24} which we call \emph{compressible subsegments}.
Hence, in order to shrink the compressible columns (rows), we need to shrink the compressible subsegments. For that, we apply the \emph{segment coloring and shrinking} subroutine to compress all compressible columns and rows (see \Cref{lemma:parallel-shrinking}).
Since we cannot shrink a segment to length $0$, we add the preceding incompressible node to each compressible subsegment and shrink the resulting subsegments to $1$. We proceed as follows.
First, we apply the subroutine on all horizontal compressible subsegments in parallel.
Then, we repeat the procedure with all vertical compressible subsegments.

\begin{theorem}
\label{th:incompressible_tree}
    After an $O(\log n)$-round preprocessing (w.h.p.), and given the incompressible nodes, the \emph{incompressible tree} algorithm reduces the initial tree $T_I$ to its incompressible form $i(T_I)$ in $O(\log n)$ rounds.
\end{theorem}

\begin{proof}
    Since we apply the segment coloring and shrinking subroutine on all compressible subsegments in parallel, all compressible subsegments of the same maximal sequence of consecutive compressible columns (rows) shrink at the same speed.
    This implies that the incompressible columns (rows) do not change, i.e., no node can enter or leave an incompressible column (row).
    These two facts imply that there cannot be a collision ---neither within the compressible columns (rows) nor within the incompressible columns (rows).
    Since no collisions can occur, the runtime follows from \Cref{lemma:parallel-shrinking}.
\end{proof}

Note that an incompressible tree contains at most $O(k^2)$ nodes since each column and row must contain at least one of the $k$ turning points.
Hence, by running the incompressible tree algorithm prior to the BFS shrinking algorithm (see \Cref{theo:shrinking-tree}), we can improve the runtime.

\begin{corollary}
    Given the incompressible nodes, the \emph{incompressible tree} algorithm combined with the \emph{BFS shrinking} algorithm reduces the initial tree to a single node in $O(\log n + k \log k)$ rounds w.h.p.
\end{corollary}

\noindent\textbf{Incompressible node computation}.
If we do not have the incompressible nodes given, we need to compute them before applying the \emph{incompressible tree} algorithm.
First, note that given an arbitrary node $u \in S$, we can apply the spatial PASC algorithm to compute for each node $v \in S$ its relative position to $u$, i.e., $r(u,v)$.

\begin{corollary}
\label{lem:relative_position}
    Let $u \in S$.
    The \emph{spatial PASC algorithm} computes $r(u,v)$ for each $v \in S$ in $O(\log n)$ rounds.
\end{corollary}

We now explain how to compute all incompressible nodes and with that all compressible subsegments.
Recall that compressible columns (rows) do not contain any turning points.
This is the case if the horizontal (vertical) distance of the column (row) to all turning points is not $0$.
Hence, we simply iterate through all turning points $u \in \mathit{TP}$ and apply the spatial PASC algorithm such that each node $v \in S$ can check whether its column (row) contains any turning points, i.e., whether for all $u \in \mathit{TP}$, $r_H(u,v) \neq 0$ ($r_V(u,v) \neq 0$).

In order to iterate through all turning points, we utilize the \emph{election primitive}\footnote{Note that the election primitive is not a leader election algorithm. It utilizes a leader node to elect a node from a given set. In case we apply the primitive in a segment, we can use the first node for that.} for trees of \cite{DBLP:conf/podc/PadalkinS24} which requires a single round.

\begin{lemma}
\label{lem:compressible_segments}
    We compute all incompressible nodes and with that all compressible segments in $O(k \log n)$ rounds.
\end{lemma}

\begin{proof}
    The correctness follows from \Cref{lem:relative_position}.
    Also by \Cref{lem:relative_position}, each iteration requires $O(\log n)$ rounds.
    We need $k$ iterations (one for each turning point).
\end{proof}

\begin{corollary}
    Given no additional assumptions, after an $O(\log n)$-round preprocessing (w.h.p.), the \emph{incompressible tree} algorithm reduces the initial tree to its incompressible shape in $O(k \log n)$ rounds.
\end{corollary}

Note that in case no additional assumptions are given, i.e., the incompressible nodes are not given, the \emph{incompressible tree} algorithm does not help to reduce the runtime of the BFS shrinking algorithm since their runtimes match.

\subsection{Topological Reduction}\label{subsec:topo:eqv:tree}

In this section, we employ both shrinking and growth operations to achieve the reduction objectives. We begin by presenting the \emph{target tree} algorithm, which reduces an arbitrary tree $T_I$ to another topologically equivalent tree $T_F$.

Recall that $\length{s^I_i}$ ($\length{s^F_i}$) denote the length of segment $s_i$ in the initial (target) tree.
We assume that for each $i$, $\length{s^F_i} \leq \length{s^I_i}$.
Furthermore, we assume that initially, each segment $s_i$ stores a binary representation of $\length{s^F_i}$ (see \Cref{sec:prelim}) which is possible due to the first assumption.  
We also assume that we have a leader node given and that the nodes share a common compass orientation and chirality.
Otherwise, we establish these assumptions in a preprocessing phase (see \Cref{sec:prelim}).

Note that $\length{s^F_i} \leq \length{s^I_i}$ for all $i$ is not sufficient for $T_F$ to be a topologically equivalent tree of $T_I$ (see \Cref{fig:topological-eqv-eg}).
We are able to test whether $T_I$ and $T_F$ are topologically equivalent.

\begin{figure}[tbp] 
    \centering
    \includegraphics[scale=1]{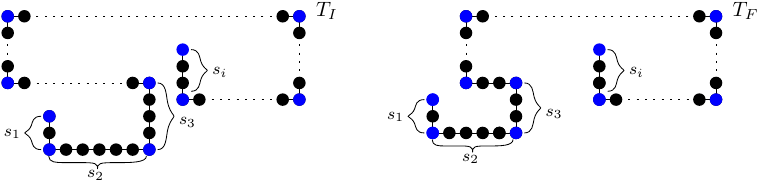}
    \caption{Even though every segment $s_i$ in the target tree $T_F$ satisfies $\length{s^F_i} \leq \length{s^I_i}$, this condition is not sufficient to guarantee that $T_F$ is topologically equivalent to $T_I$. In order to transform $T_I$ into $T_F$, the algorithm must also preserve the relative order of the turning points (highlighted in blue).}
    \label{fig:topological-eqv-eg}
\end{figure}

\noindent\textbf{Topological equivalence test}.
For each turning point $u$, we apply the spatial PASC algorithm on both $T_I$ and $T_F$, so that all other turning points can compare their relative positions to $u$ in $T_I$ and $T_F$, respectively.
However, we are not able to apply the subroutine on $T_F$ directly.
Instead, we will simulate $T_F$ in $T_I$ by identifying a unique node in $T_I$ for each node in $T_F$.
All other nodes will simply forward the signals between their neighbors.
This is effectively the same as if those nodes were not present in the execution.
Note that each turning point in $T_F$ already has a corresponding node in $T_I$.
Hence, we only have to identify $\length{s^F_i} - 2$ further nodes in each segment $s_i$.
We use the PASC algorithm to identify the first $\length{s^F_i} - 2$ segment nodes in each $s_i$.

\begin{lemma}
\label{lem:sim}
    Let $\cal A$ be any algorithm that does not apply any growth or shrinking operations.
    If for each $i$, $\length{s^F_i} \leq \length{s^I_i}$, $T_I$ can simulate $\cal A$ on $T_F$ after a preprocessing time of $O(\log n)$ rounds.
    
\end{lemma}

\begin{proof}
    The turning points are already given. The PASC algorithm identifies $\length{s^F_i} - 2$ segment nodes in each segment $s_i$ and requires $O(\log n)$ rounds (see \Cref{lem:pasc:segment}).
\end{proof}

The \emph{topological equivalence subroutine} proceeds as follows.
We iterate through all turning points by utilizing the election primitive of \cite{DBLP:conf/podc/PadalkinS24} as before.
For each turning point $u$, we apply the spatial PASC algorithm on $T_I$ and $T_F$.
This allows each of the other turning points, to compare its relative position to $u$ in $T_I$ and $T_F$. Finally, after comparing all relative positions, we perform a notification round where we check whether there is a turning point that has identified a violation, i.e., a difference of its relative position to another turning point. For that, we first establish a circuit connecting all nodes. Then, each turning point that has identified a violation beeps on that circuit.
$T_F$ is topologically equivalent to $T_I$ if and only if there was no beep.

\begin{lemma}\label{lem:topo:eqv:subroutine}
    The \emph{topological equivalence} subroutine checks whether $T_F$ is topologically equivalent to $T_I$ in $O(k \log n)$ rounds.
\end{lemma}

\begin{proof}
    By \Cref{lem:sim}, we are able to perform the spatial PASC algorithm on $T_F$.
    The correctness of the topological equivalence subroutine follows from \Cref{lem:relative_position}.
    By \Cref{lem:sim}, the preprocessing for $T_F$ requires $O(\log n)$ rounds.
    By \Cref{lem:relative_position}, each iteration takes $O(\log n)$ rounds.
    We must perform $k$ iterations (one for each turning point).
    The final notification round only adds a single round. Overall, the topological equivalence subroutine requires $O(k \log n)$ rounds.
\end{proof}

\noindent\textbf{Target tree reduction}. We start by outlining the challenges of reducing the initial tree to a target tree.
The first idea one may have is to apply the \emph{segment coloring and shrinking} subroutine (see \Cref{lemma:parallel-shrinking}) on all segments in parallel until each segment has reached its target length.
However, even if $T_I$ and $T_F$ are topologically equivalent, this approach does not guarantee that intermediate trees stay topologically equivalent which may cause collisions.

To always avoid collisions, we will apply the \emph{segment coloring and shrinking} subroutine on the compressible subsegments instead of the tree segments.
By uniformly changing the length of the compressible subsegments belonging to the same sequence of consecutive compressible columns (rows), we can make sure that all intermediate trees are topologically equivalent and with the guarantee that no collisions can occur (see \Cref{fig:invalid-transformation}).
\begin{figure}[tbp]
    \centering
    \includegraphics[scale=1]{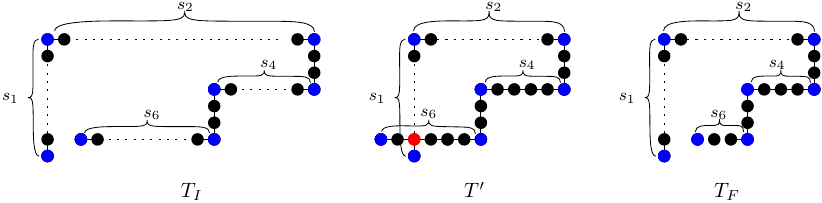}
    \caption{Applying parallel shrinking on horizontal segments ($s_2$, $s_4$, and $s_6$) to transform the initial tree $T_I$ into the target tree $T_F$, does not ensure that intermediate trees (e.g., $T'$) remain topologically equivalent. Instead, it might result in collisions (as indicated by the red node where segment $s_6$ collides with $s_1$).}
    \label{fig:invalid-transformation}
\end{figure}
Let $\length{c^I_j}$ ($\length{c^F_j}$) denote the length of the compressible subsegment $c_j$ in the initial (target) tree.
However, the assumption that for each $i$, $\length{s^F_i} \leq \length{s^I_i}$ does not imply that for each $j$, $\length{c^F_j} \leq \length{c^I_j}$, i.e., some compressible subsegments grow (see \Cref{fig:compressible-seg}). It is even possible that $\length{c^F_j}/\length{c^I_j} = \Omega(n)$ which makes it impossible to store $\length{c^F_j}$ within $c_j$.

\begin{figure}[tbp]
    \centering
    \includegraphics[scale=1.1]{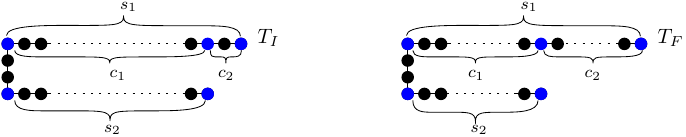}
    \caption{An example where shrinking the segment $s_2$ in the initial tree $T_I$ doesn't imply that the compressible subsegments $c_1$ and $c_2$ shrink accordingly in the target tree $T_F$. In fact, in $T_F$, the subsegment $c_2$ is larger than its initial length in $T_I$.}
    \label{fig:compressible-seg}
\end{figure}

In the following, we will explain how to deal with all these challenges.
We first outline our \emph{target tree} algorithm. It consists of three phases: \emph{compressible subsegment computation} phase, \emph{growth} phase, and \emph{shrinking} phase.
In the \emph{compressible subsegment computation} phase, we compute the lengths of all compressible subsegments of $T_I$ and $T_F$. Then, we grow or shrink all compressible subsegments in the other two phases.
In the \emph{growth} phase, we grow all compressible subsegments $c_j$ with $\length{c^F_j} > \length{c^I_j}$ and in the \emph{shrinking} phase, we shrink all compressible subsegments $c_j$ with $\length{c^F_j} < \length{c^I_j}$.

\medskip
\noindent\emph{Compressible subsegment computation} phase.
In this phase, we compute the lengths of all compressible subsegments of $T_I$ and $T_F$, respectively (see \Cref{fig:computation-phase}).
For that, we first compute all compressible columns and rows of $T_I$ and $T_F$ (see \Cref{lem:compressible_segments}).
Note that we need to simulate the subroutine for $T_F$ (see \Cref{lem:sim}).

Then, we apply the PASC algorithm on each compressible subsegment $c^I_j$ and $c^F_j$ in parallel to compute its length ($+1$) which we store within the segment itself.
Consider the compressible subsegments $c_j$ of a segment $s_i$.
For each $c_j$, $s_i$ now stores $\length{c^I_j} + 1$ and $\length{c^F_j} + 1$.
However, these lengths are not necessarily aligned since the incompressible nodes may differ between $T_I$ and $T_F$.
Hence, we apply the transfer primitive to shift $\length{c^I_j}$ and $\length{c^F_j}$ of each $c_j$ to a subsegment $c'_j$ of $s_i$.

\begin{figure}[tbp]
    \centering
    \includegraphics[scale=1.2]{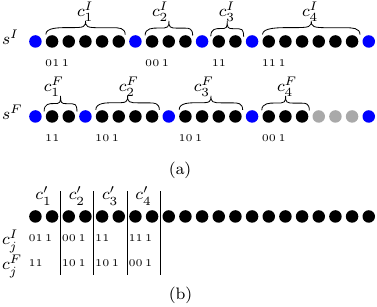}
    \caption{
        An illustration of the compressible subsegment computation phase: (a) depicts the computation of compressible subsegments (black nodes) along with their lengths, while blue nodes represent incompressible ones. (b) shows the result after aligning the lengths using the transfer primitive.
    }
    \label{fig:computation-phase}
\end{figure}

\begin{lemma}
\label{lem:ttt:compressible_subsegment_computation}
    In the \emph{compressible subsegment computation} phase, the \emph{target tree} algorithm computes $\length{c^I_j} + 1$ and $\length{c^F_j} + 1$ for each $c_j$ and stores it in a segment $c'_j$ in $O(k \log n)$ rounds.
\end{lemma}
\begin{proof}
    The correctness follows by construction.
    We only need to make sure that we have enough nodes for all $c'_j$'s.
    Each $c'_j$ needs at least $\left\lceil \log \left( \max \{\length{c^I_j}, \length{c^F_j}\} + 2 \right) \right\rceil/c$ nodes to store the lengths where $c$ is the number of bits per node and value.
    For $c \geq 2$, we obtain
    \begin{align*}
        \sum_{c_j \subseteq s_i} \length{c'_j}
        &= \sum_{c_j \subseteq s_i} \left\lceil \log \left( \max \{\length{c^I_j}, \length{c^F_j}\} + 2 \right)\right\rceil/c \\
        &\leq \frac{1}{c} \sum_{c_j \subseteq s_i} \left(\max\{\length{c^I_j}, \length{c^F_j}\} + 2 \right) \\
        &\leq \frac{1}{c} \sum_{c_j \subseteq s_i} \left(\length{c^I_j} + \length{c^F_j} + 2 \right) \\
        &\leq \frac{1}{c} \left( \sum_{c_j \subseteq s_i} \length{c^I_j} + (k_i - 1) + \sum_{c_j \subseteq s_i} \length{c^F_j} + (k_i - 1) \right) \\
        &\leq \frac{1}{c} \left( \length{s^I_i} + \length{s^F_i} \right) \\
        &\leq \frac{2}{c} \cdot \length{s^I_i} \\
        &= \length{s^I_i}
    \end{align*}
    where $k_i$ denotes the number of incompressible nodes in $s_i$. Therefore, we have enough memory space to store all values.

    By \Cref{lem:sim}, the preprocessing for the simulation of $T_F$ requires $O(\log n)$ rounds.
    By \Cref{lem:compressible_segments}, the computation of all compressible subsegments requires $O(k \log n)$ rounds.
    By \Cref{lem:pasc:segment}, the computation of the lengths requires $O(\log n)$ rounds.
    The execution of the transfer primitive requires $O(k_i \log n) = O(k \log n)$ rounds (see \Cref{sec:prelim}) since we have $k_i - 1$ compressible subsegments and the lengths are bounded by $\length{s^I_i} \leq n$.
\end{proof}

\medskip
\noindent\emph{Growth and shrinking} phases.
In the growth (shrinking) phase, we grow (shrink) all compressible subsegments $c_j$ with $\length{c^F_j} > \length{c^I_j}$ ($\length{c^F_j} < \length{c^I_j}$).
Observe that it does not matter which nodes exactly perform a growth (shrinking) operation in a segment as long as the number of operations is the same.
The idea is therefore to first compute a subsegment $c''_j$ for each $c_j$ and then to perform the growth (shrinking) operations of $c_j$ in $c''_j$. Note that we cannot use $c'_j$ for $c_j$ since it may not have enough nodes, i.e., $\length{c'_j} < \length{c^I_j} + 1$.

We split the phase into two subphases: \emph{segmentation }subphase and \emph{transformation} subphase.
In the \emph{segmentation} subphase, we compute a subsegment $c''_j$ for each $c_j$ and in the \emph{transformation} subphase, we apply the growth (shrinking) operations.
In both subphases, we will only consider the $c_j$'s relevant for the phase, i.e., the $c_j$'s with $\length{c^F_j} > \length{c^I_j}$ ($\length{c^F_j} < \length{c^I_j}$) for the growth (shrinking) phase.
Note that each $c'_j$ is able to compare $\length{c^F_j}$ and $\length{c^I_j}$ to determine in which phase it is participating.

\noindent\emph{Segmentation subphase}.
In this phase, we compute a subsegment $c''_j$ for each $c_j$ (see \Cref{fig:growth-phase,fig:shrinking-phase}).
The subsegment has to satisfy three properties.
First, it must be long enough to perform all operations at the same rate as $c_j$.
This is the case if $\length{c''_j} \geq \length{c^I_j} + 1$.
Second, it must be long enough to store $\length{c^I_j}$ and $\length{c^F_j}$.
This is the case if $\length{c''_j} \geq \length{c'_j} \geq \left\lceil \log \left( \max \{\length{c^I_j}, \length{c^F_j}\} + 2 \right) \right\rceil/c$ where $c$ is the number of bits per node and value.
We will use $\length{c''_j} = \max\{ \length{c^I_j}, \length{c^F_j} \} + 1$ which satisfies the first two properties.
Third, all subsegments $c''_j$ are pairwise disjoint. The subphase iterates through all $c'_j$ by utilizing the election primitive of \cite{DBLP:conf/podc/PadalkinS24}.
In each iteration, we proceed as follows.
Let $s'_i$ be the subsegment without all already computed subsegments.
Initially, $s'_i = s_i$.
First, we compute $c''_j$ by applying the PASC algorithm on $s'_i$ to identify the first $\max\{ \length{c^I_j}, \length{c^F_j} \} +1$ nodes.
Then, we transfer $\length{c^F_j}$ and $\length{c^I_j}$ from $c'_j$ to $c''_j$.
Finally, we remove $c''_j$ from $s'_i$ and proceed to the next iteration.
We proceed to the next subphase once all $s_i$ have computed their $c''_j$'s.

\begin{figure}[htbp!]
    \centering
    \includegraphics[scale=1.2]{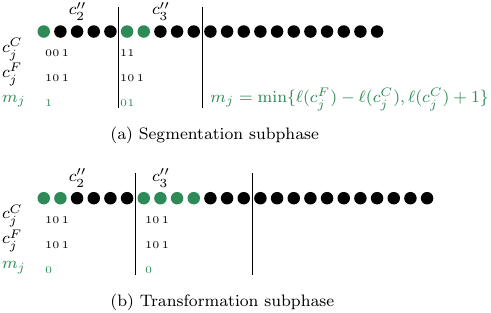}
    \caption{
        Growth phase.
        The green nodes in (a) perform the growth operations.
        The green nodes in (b) show the result of these operations.
    }
    \label{fig:growth-phase}
\end{figure}

\begin{figure}[htbp!]
    \centering
    \includegraphics[scale=1.2]{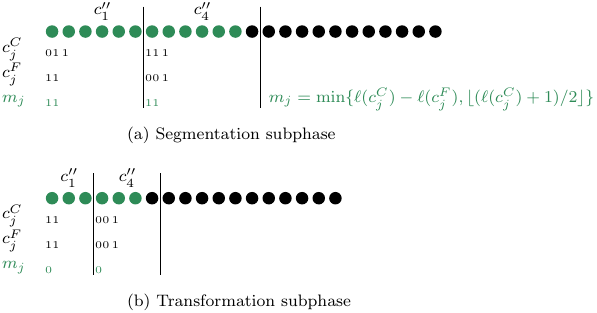}
    \caption{
        Shrinking phase.
        The green nodes in (a) perform the shrinking operations.
        The green nodes in (b) show the result of these operations.
    }
    \label{fig:shrinking-phase}
\end{figure}

\begin{lemma}
\label{lem:ttt:segmentation}
    In the \emph{segmentation} subphase, the \emph{target tree} algorithm computes a disjoint subsegment $c''_j$ with length $\max\{ \length{c^I_j}, \length{c^F_j} \} + 1$ and stores $\length{c^F_j}$ and $\length{c^I_j}$ in it in $O(k \log n)$ rounds.
\end{lemma}

\begin{proof}
    Since $s'_i$ always stays connected, the correctness follows by construction.
    We only need to make sure that we have enough nodes for all $c''_j$'s.
    For that, we consider the subphase separately for the growth and shrinking phase.
    Note that $s_i$ has at least $\length{s^I_i}$ nodes in both executions of the subphase since we only add additional nodes in between them.
    
    First, consider the growth phase.
    Since $\length{c''_j} = \max\{ \length{c^I_j}, \length{c^F_j} \} + 1 = \length{c^F_j} + 1$,
    we obtain $$\sum_{\substack{c_j \subseteq s_i,\\ \length{c^F_j} > \length{c^I_j}}} \length{c'_j} = \sum_{\substack{c_j \subseteq s_i,\\ \length{c^F_j} > \length{c^I_j}}} \left( \length{c^F_j} + 1 \right) \leq \length{s^F_i} \leq \length{s^I_i}.$$
    Next, consider the shrinking phase.
    Since $\length{c''_j} = \max\{ \length{c^I_j}, \length{c^F_j} \} + 1 = \length{c^I_j} + 1$,
    we obtain $$\sum_{\substack{c_j \subseteq s_i,\\ \length{c^F_j} < \length{c^I_j}}} \length{c'_j} = \sum_{\substack{c_j \subseteq s_i,\\ \length{c^F_j} < \length{c^I_j}}} \left( \length{c^I_j} + 1 \right) \leq \length{s^I_i}.$$
    Hence, in both phases, we have enough nodes for all $c''_j$'s.

    By \Cref{lem:pasc:segment}, the PASC algorithm requires $O(\log n)$ rounds.
    The transfer of the lengths require $O(\log n)$ rounds (see \Cref{sec:prelim}).
    All other steps only require a single round.
    Hence, each iteration requires $O(\log n)$ rounds.
    Since we need $O(k_i) = O(k)$ iterations, the runtime follows.
\end{proof}

\medskip
\noindent\emph{Transformation subphase}.
In this subphase, we apply the growth (shrinking) operations on all $c''_j$ in parallel (see \Cref{fig:growth-phase,fig:shrinking-phase}).
In each iteration, subsegment $c''_j$ proceeds as follows.
Let $\length{c^C_j}$ denote the current length of $c_j$.
Initially, $\length{c^C_j} = \length{c^I_j}$.
First, it computes the minimum $m_j$ of the number of necessary and possible growth (shrinking) operations.
In the growth phase, this is $m_j = \min\{ \length{c^F_j} - \length{c^C_j}, \length{c^C_j} + 1 \}$ and in the shrinking phase, this is $m_j = \min\{ \length{c^C_j} - \length{c^F_j}, \lfloor (\length{c^C_j} + 1)/2 \rfloor \}$.
Then, it updates $\length{c^C_j}$ to $\length{c^C_j} + m_j$ in the growth phase or $\length{c^C_j} - m_j$ in the shrinking phase.
Next, it marks the number of nodes necessary to perform all operations.
For that, it applies the PASC algorithm to identify the first $m_j$ nodes in the growth phase or the first $2m_j$ nodes in the shrinking phase.
We mark all nodes with a distance less than $m_j$ in the growth phase and $2m_j$ in the shrinking phase.

\begin{lemma}
\label{lem:ttt:transformation}
    In the \emph{transformation} subphase, the \emph{target tree} algorithm grows (shrinks) all compressible subsegments $c_j$ with $\length{c^F_j} > \length{c^I_j}$ ($\length{c^F_j} < \length{c^I_j}$) to a length of $\length{c^F_j}$ in $O(\log^2 n)$ rounds.
\end{lemma}

\begin{proof}
    The correctness follows from the same arguments as in the proof of \Cref{th:incompressible_tree}.

    The computation of $m_j$ and update of $c^C_j$ require $O(1)$ rounds.
    By \Cref{lem:pasc:segment}, the PASC algorithm requires $O(\log n)$ rounds.
    The growth/shrinking operations require $O(1)$ rounds (for the shrinking phase, see \Cref{lemma:parallel-shrinking}).
    The transfer primitive requires $O(\log n)$ rounds.
    Overall, each iteration of the transformation phase requires $O(\log n)$ rounds.

    It remains to bound the number of iterations.
    For that, note that $\length{c^F_j} \leq n$ since otherwise, $\length{s^F_i} > \length{c^F_j} > n \geq \length{s^I_i}$ which is a contradiction to our assumptions.
    Since we can double/half the nodes of $c_j$ in each iteration, $O(\log n)$ iterations suffice.
\end{proof}

Finally, we apply the operations.
In the \emph{growth} phase, each marked node performs a growth operation.
In the \emph{shrinking} phase, we apply one iteration of the \emph{segment coloring and shrinking} subroutine on the marked nodes.
Note that the growth operations may split up the binary representation of the stored values.
In this case, we deal with them in the same way as with the excess nodes during the simulation of $T_F$ (see \Cref{lem:sim}):
The newly grown nodes do not participate in the computation of $m_j$ (or any other future computation, e.g., in the next segmentation subphase) and simply forward the signals which has the same effect as removing them again.
However, they do participate in the final PASC algorithm to mark nodes to grow. We have a similar problem in the \emph{shrinking} phase.
If we apply a shrinking operation on two nodes that store $c$ bits, the absorbing node has to store $2c$ bits.
Since we may need to apply up to $O(\log n)$ shrinking operations, the absorbing node can end up with $cn$ bits which it is not capable of storing.
In order to resolve this issue, we redistribute the stored values within the subsegment by using the transfer primitive. We proceed to the next phase (if we are in the \emph{growth} phase) or terminate the algorithm (if we are in the \emph{shrinking} phase) once $m_j = 0$ for all $c_j$ in all $s_i$.

\begin{theorem}\label{the:target-tree}
    After an $O(\log n)$-round preprocessing (w.h.p.), the \emph{target tree} algorithm reduces $T_I$ to $T_F$ in $O(k \log n + \log^2 n)$ rounds.
\end{theorem}

\begin{proof}
    The correctness and runtime follow from \Cref{lem:ttt:compressible_subsegment_computation,lem:ttt:segmentation,lem:ttt:transformation}.
\end{proof}

\section{The Adjacency Model}\label{sec:adjacency} 
In this section, we consider the \textsc{Single Node Reduction} problem in the \emph{adjacency model}, where the initial shape is any connected shape. We begin by electing a unique leader in $O(\log n)$ rounds w.h.p. (see \Cref{sec:prelim}). We also compute compass alignment and a common chirality among nodes, which is similarly completed in $O(\log n)$ rounds w.h.p. Once the leader election and compass alignment are complete, all nodes begin the main shrinking process simultaneously, guided by a signal from the leader. The reduction is obtained through a distributed simulation of the \emph{elimination} algorithm proposed in \cite{DBLP:conf/algosensors/AlmalkiGM24}, referred to hereafter as the \emph{shape reduction} algorithm. Specifically, for any connected initial shape $S_I$ consisting of $C$ columns and $R$ rows, the \emph{shape reduction} algorithm applies the spatial PASC (\Cref{lem:pasc:spatial}) to first partition the shape based on the parity of its columns and rows, alternating between shrinking columns and rows until $S_I$ is reduced to a single node. 

\medskip
\noindent\textbf{Shape reduction.}
Consider any connected shape $S_I$, in the adjacency model, where edges are added between all adjacent nodes when they become adjacent~\cite{DBLP:conf/algosensors/AlmalkiGM24}. It is important to note that this process has the \emph{neighbor handover} property, where the node $u_j$ absorbs $u_{j+1}$ and inherent any node connected to $u_{j+1}$.

Let the shape $S_I$, consists of $C$ columns and $R$ rows. Let $c_j \in {C}$ be a column indexed by $j$ from west to east, and let $r_i \in {R}$ be a row indexed by $i$ from south to north. Each column $c_j\in C$ and row $r_i \in {R}$ has two partition sets: primary and secondary. 
Between columns, the primary partition set of a column $c_j$ is connected to the secondary partition set of its predecessor $c_{j-1}$. Similarly, the secondary partition set of a column $c_j$ is connected to the primary partition set of its predecessor $c_{j-1}$. 
Within a column, the primary partition set of a row $r_i$ is connected to the primary partition set of its predecessor $r_{i-1}$. The secondary partition set of a row $r_i$ is connected to the secondary partition set of its predecessor $r_{i-1}$.
This configuration forms two disjoint circuits---primary and secondary circuits---for each column $c_j$, in these circuits, the partition sets alternate between primary and secondary. In the preprocessing phase, we elect a reference node $u_0$ in a column $c_{j'}$. This reference node propagates signals along the circuits to compute the parity of the distance from column $c_{j'}$ to each column $c_j \in C$. In particular, if a column receives the signal via its secondary partition set, it is assigned odd parity; if it receives the signal via its primary partition set, it is assigned even parity.

After partitioning all the columns into even and odd sets, for every even-parity column $c_j \in C$ (blue), with an odd-parity predecessor $c_{j-1}$ (green), the even-parity column $c_j$ absorbs its edge $(c_{j-1},c_j)$ towards $c_j$. In doing so, the primary circuit of the predecessor of $c_{j-1}$ (i.e., $c_{j-2}$) connects to the secondary circuit of $c_{j}$, and the secondary circuit of $c_{j-2}$ connects to the primary circuit of $c_{j}$. After shrinking all odd-parity columns, the spatial PASC algorithm is recomputed on the set of all rows $R$, partitioning them into even and odd sets. The odd-parity rows shrunk in the same manner analogous to the columns. This alternating process of column-shrinking followed by row-shrinking continues until the shape reduces to a single node. The final remaining node is able to detect that it is the last node, allowing it to terminate the shrinking process (see \Cref{fig:spatial_coloring}).

\begin{figure}[htbp!]
    \centering
    \includegraphics[scale=0.8]{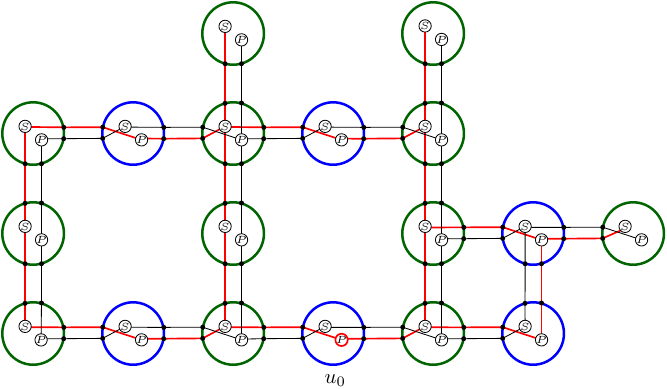}
    \caption{The application of spatial PASC on a shape $S$, where each adjacent pair of nodes is connected by an edge.}
    \label{fig:spatial_coloring}
\end{figure}

\begin{theorem}\label{theo:shrink-S}
    After an $O(\log n)$-round preprocessing (w.h.p.), 
    the \emph{shape reduction} algorithm shrinks any connected shape $S_I$ to a single node in $O(\log \max\{C, R\}) = O(\log n)$ rounds.
\end{theorem}

\begin{proof}
Let $S_I$ be a connected shape consisting of $C$ columns and $R$ rows. The distributed shrinking via \emph{spatial PASC} algorithm alternates between shrinking columns and rows. In the first step, all columns $\{c_1, c_2, \ldots, c_C\}$ are partitioned into even-parity and odd-parity sets using \emph{spatial PASC} (see~\cite{DBLP:journals/nc/PadalkinSW24}), and the odd-parity columns are shrunk. This reduces the number of columns to $C / 2$, while the number of rows remains $R$. In the next step, the rows $\{r_1, r_2, \ldots, r_R\}$ are similarly partitioned in the same way which reduces the number of rows to $R / 2$, while the reduced columns remain unchanged. This alternation between columns and rows ensures that shrinking proceeds.
After $r$ rounds, the number of columns and rows is halved $C_r = C / 2^{\lceil r / 2 \rceil},$ $ R_r = R / 2^{\lfloor r / 2 \rfloor}$. The process continues until $C_r = 1$ and $R_r = 1$. Since the maximum number of rounds required is $\log C$ for columns and $\log R$ for rows, the total time complexity of the algorithm is $O(\log \max \{C + R\})$.
\end{proof}

\section{Open Problems}\label{sec:conclusion}
Several open problems remain to be addressed. One problem is to study the symmetric shrinking mechanisms that avoid reliance on leaders or predefined orientations. Another problem is encoding the target shape and transforming between shapes, particularly under memory constraints, for example, growing from a single node to a target shape or transforming a larger shape to a smaller one, or a smaller to a larger. Mixed models that allow both shrinking and growing also raise questions about how to preserve memory. Also, exploring alternative communication methods, such as local visibility, offers the potential for improving coordination in these systems.
An exciting broader research direction opened by this work is the integration of rapid communication schemes, such as reconfigurable circuits, into fast, parallel transformation processes, potentially combining size-preserving and size-changing operations.

\bibliography{bibliography}

\end{document}